\documentclass[a4paper]{article}
\usepackage[latin1]{inputenc}
\usepackage{fullpage}
\usepackage{lmodern}
\usepackage{subfigure}
\usepackage{amssymb,amsmath,amsthm}
\usepackage{graphicx,color,xspace}
\usepackage{hyperref}

\newenvironment{di}{\list{$\bullet$}{\itemsep0pt\parsep0pt}}{\endlist}

\newtheorem{theorem}{Theorem}
\newtheorem{lemma}[theorem]{Lemma}

\newcommand{\SITES}{\mathcal{S}}
\newcommand{\R}{\mathbb{R}}
\newcommand{\conv}{\operatorname{CH}}
\newcommand{\MA}{\mathfrak{M}}
\newcommand{\DIST}{\Psi}
\newcommand{\UDIST}{\Phi}
\newcommand{\fpvd}{\mathfrak{F}}
\newcommand{\REG}{\mathcal{R}}
\newcommand{\TREE}{\mathcal{T}}
\newcommand{\FI}{\fpvd(\SITES_{1})}
\newcommand{\FII}{\fpvd(\SITES_{2})}
\newcommand{\FS}{\fpvd(\SITES)}
\newcommand{\s}{^{\ast}}
\newcommand{\cyl}{\mathfrak{C}}
\newcommand{\pocket}{\mathcal{P}}
\newcommand{\etal}{{et al.}\xspace}

\definecolor{complaincolor}{rgb}{1.0,0.7,0.3}
\newcommand{\complain}[1]{\textcolor{complaincolor}{\bfseries#1}\marginpar{\textcolor{complaincolor}{!!!}}}
\renewcommand{\complain}[1]{}
\newcommand{\SL}[1]{\complain{Sylvain says: #1}}

\definecolor{orange}{rgb}{1,0.5,0}
\definecolor{blue}{rgb}{0,0,1}
\definecolor{red}{rgb}{1,0,0}
\definecolor{green}{rgb}{0,0.6,0}
\def\NOTE#1#2#3{{\begin{quote}\marginpar[\hfill\textcolor{#1}{#2}]{\textcolor{#1}{#2}}{\textcolor{#1}{\textsf{[\![{#3}]\!]}}}\end{quote}}}

\iffalse
\def\referee#1{\NOTE{orange}{$\circledR$}{#1}}
\def\respond#1{\NOTE{blue}{\textcircled{\textsc{a}}}{Authors:~{#1}}}
\else
\def\referee#1{}
\def\respond#1{}
\fi

\let\leq\leqslant
\let\geq\geqslant

\title{Farthest-Polygon Voronoi Diagrams%
  \thanks{This research was supported by 
    the INRIA équipe associée \emph{KI}, the Brain
    Korea~21 Project, the School of Information Technology, KAIST, 
    and the Korea Science and Engineering
    Foundation Grant~R01-2008-000-11607-0 funded by the Korean
    government.}}

\author{%
  Otfried Cheong%
  \thanks{Dept.~of Computer Science, KAIST, Daejeon, Korea.
    \{otfried,mira\}@tclab.kaist.ac.kr.}
  \and
 Hazel Everett\thanks{Université Nancy 2, LORIA, Nancy, 
    France. Hazel.Everett@loria.fr.}
  \and
	 Marc Glisse\thanks{INRIA Saclay -- Île-de-France, Orsay, France.
	Marc.Glisse@inria.fr.}
  \and
  Joachim Gudmundsson%
  \thanks{National ICT Australia Ltd., Sydney, Australia.
    joachim.gudmundsson@nicta.com.au. National ICT Australia is funded
    through the Australian Government's Backing Australia's Ability
    initiative, in part through the Australian Research Council.}
  \and
  Samuel Hornus\thanks{INRIA Nancy Grand-Est, LORIA, Nancy, 
    France. Firstname.Name@inria.fr.}
  \and
  Sylvain Lazard\footnotemark[6]
  \and
  Mira Lee\footnotemark[2]
  \and
  Hyeon-Suk Na%
  \thanks{School of Computing, Soongsil University, Seoul, Korea.
    hsnaa@ssu.ac.kr.}}

\begin{document}

\maketitle

\begin{abstract}  
   Given a family of $k$ disjoint connected polygonal sites in general
  position and of total complexity~$n$, we consider the farthest-site Voronoi
  diagram of these sites, where the distance to a site is the distance to a
  closest point on it. We show that the complexity of this diagram is $O(n)$,
  and give an $O(n \log^{3} n)$ time algorithm to compute it. We also prove a
  number of structural properties of this diagram. In particular, a Voronoi
  region may consist of $k-1$ connected components, but if one component is
  bounded, then it is equal to the entire region. 
  \end{abstract}

\section{Introduction}

Consider a family $\SITES$ of geometric objects, called sites, in
the plane.  The farthest-site Voronoi diagram of $\SITES$ subdivides
the plane into regions, each region associated with one site $P\in
\SITES$, and containing those points $x\in \R^{2}$ for which $P$ is
the farthest among the sites of~$\SITES$.

While closest-site Voronoi diagrams have been studied
extensively~\cite{ak-vd-00}, their farthest-site cousins have received
somewhat less attention. For the case of (possibly intersecting) line segment
sites, Aurenhammer~\etal~\cite{adk-flsvd-06} recently presented an $O(n \log
n)$ time algorithm to compute their farthest-site diagram.
%The case of (possibly intersecting) line segment sites was only solved
%recently by Aurenhammer~\etal~\cite{adk-flsvd-06}; they gave an $O(n \log
%n)$ time algorithm to compute the diagram for $n$ line segments.

Farthest-site Voronoi diagrams have a number of important applications.
Perhaps the most well-known one is the problem of finding a smallest disk that
intersects all the sites~\cite{ahiklmps-fcvdr-01}. This disk can be computed
in linear time once the diagram is known, since its center is a vertex or lies
on an edge of the diagram.
%Other applications are
%finding the largest gap to be bridged between sites, or building a
%data structure to quickly report the site farthest from a given query
%point.
Another standard application is to build a data structure to quickly report
the site farthest from a given query point. 
%\complain{Samuel: I put the
%sentence back, with the additional word ``standard''. Is is an important
%application IMO: our own divide and conquer algo uses this point location
%application for building the FPVD.} \SL{Fine with me.}

We are here interested in the case of complex sites with non-constant
description complexity.  This setting was perhaps first considered by
Abellanas~\etal~\cite{ahiklmps-fcvdr-01}: their sites are finite
point sets, and so the distance to a site is the distance to the
nearest point of that site.  Put differently, they consider $n$ points
colored with $k$ different colors, and their \emph{farthest-color
Voronoi diagram} subdivides the plane depending on which color is
farthest away.  The motivation for this problem is the one mentioned
above, namely to find a smallest disk that contains a point of each
color---this is a facility location problem where the goal is to find
a position that is as close as possible to each of $k$ different types
of facilities (such as schools, post offices, supermarkets, etc.). In
a companion paper~\cite{ahiklmps-scso-01} the authors study other
color-spanning objects.

The farthest-color Voronoi diagram is easily seen to be the projection
of the upper envelope of the $k$ Voronoi surfaces corresponding to the
$k$ color classes. Huttenlocher~\etal~\cite{hks-uevsi-93} show that
this upper envelope has complexity $\Theta(nk)$ for $n$ points, and
can be computed in time $O(nk \log n)$ (see also the book by Sharir
and Agarwal~\cite[\S 8.7]{sa-dsstg-95}.
%, Section~8.7).

\referee{Referee \#1: Figure 1d: It would be helpful to indicate which of the regions belongs to which polygon.}
\respond{This is now indicated in the  caption of Fig. 1c and directly on the figure, with an arrow and a label.}
%\SL{Modified Sam's change on caption of Fig 1c for english}

Van Kreveld and Schlechter~\cite{KS05} consider the farthest-site
Voronoi diagram for a family of disjoint simple polygons.  Again, they
are interested in finding the center of the smallest disk intersecting
or touching all polygons, which they then apply to the cartographic
problem of labeling groups of islands.  Their algorithm is based on
the claim that this \emph{farthest-polygon Voronoi diagram} is an
instance of the \emph{abstract farthest-site Voronoi diagram} defined
by Mehlhorn~\etal~\cite{MMR01}---but this claim is false, since the
bisector of two disjoint simple polygons can be a closed curve, see
Fig.~\ref{fig:closed-curve}(a).  In particular, Voronoi regions can be
bounded (which is impossible for regions in abstract farthest-site
Voronoi diagrams).
\begin{figure}
  \centerline{\includegraphics{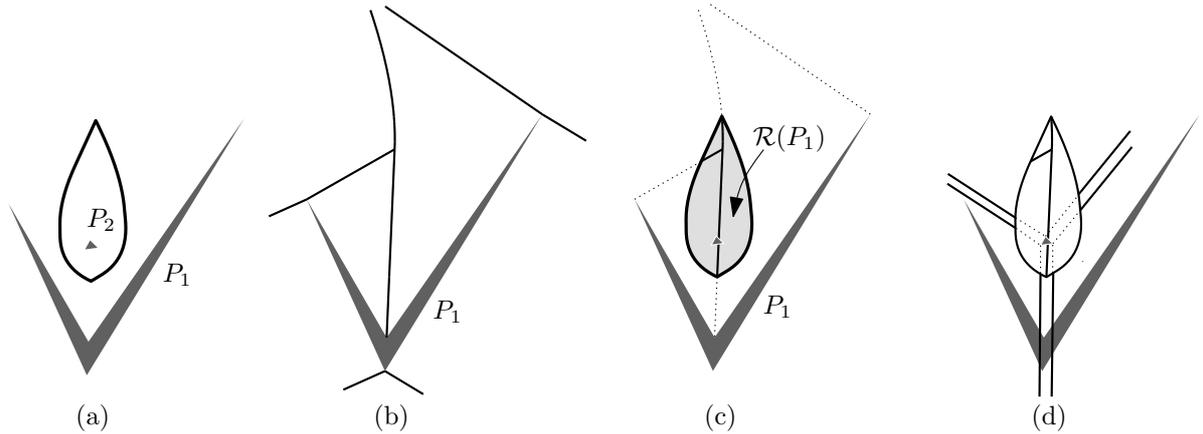}}
  \caption{(a) The bisector of two polygons can be a closed curve.
    (b) The medial axis $\MA(P_{1})$ outside of $P_{1}$.
%    (c) The gray area is the Voronoi region $\REG(P_1)$.
    (c) The Voronoi region $\REG(P_1)$, shown in light gray. 
    (d) The farthest-polygon Voronoi diagram
    $\fpvd(\{P_1, P_2\})$.}
  \label{fig:closed-curve}
\end{figure}

Note that the farthest-polygon Voronoi diagram can again be expressed
as the upper envelope of $k$ Voronoi surfaces---but this does not seem
to lead to anything stronger than near-quadratic complexity and time
bounds. 

We show in this paper that, in fact, the complexity of the farthest-polygon
Voronoi diagram of $k$ disjoint simple polygons of total complexity $n$
is~$O(n)$.
%(We actually consider %And in fact we discuss
%\emph{polygonal sites}, which are slightly more general than simple
%polygons---see Section~\ref{sec:prelim}.)  
%%      This sentence raises the question of whether the lower bound below holds for polygons or only
%%      polygonal sites. It seems reasonable to not mention here the generalization to polygonal
%%sites
We also show some structural properties of this diagram. In
particular, Voronoi regions can be disconnected, and in fact, the
region of a polygon $P$ can consist of up to $k-1$ connected
components. However, if one connected component is bounded then it is
equal to the entire region of $P$; moreover, the region is simply
connected and the convex hull of $P$ contains another polygon in its
interior.  Furthermore, the Voronoi regions consist, in total, of at
most $2k-2$ connected components, and this bound is tight.

%Voronoi regions can be disconnected, and in fact, the region of a polygon can
%consist of up to $k-1$ connected components, but then all these components
%must be unbounded. 
%However, all the Voronoi regions consist of at most $2k-2$ connected components in total. 
%If the Voronoi region of a polygon $P$ is non-empty and
%bounded, then the convex hull of $P$ contains another polygon in its interior
%and the Voronoi region of $P$ is simply connected.

Algorithms for  computing closest-site Voronoi diagrams  make use of the
fact that Voronoi regions surround and are close to their sites.
Similarly, algorithms for computing farthest-site Voronoi diagrams
make use of the unboundedness of the Voronoi regions, and often
build up regions from infinity~\cite{adk-flsvd-06}.  The difficulty
in computing farthest-polygon Voronoi diagrams is that neither of
these properties holds: Voronoi regions can be bounded, and finding
the location of these bounded regions is the bottleneck in the
computation.

We give a divide-and-conquer algorithm with running time $O(n \log^{3}
n)$ to compute the farthest-polygon Voronoi diagram.  Our key idea is
to build point location data structures for the partial diagrams
already computed, and to use parametric search on these data
structures to find suitable starting vertices for the merging step.
This idea may find applications in %to 
% I am not totally sure but I believe "in" is more adequate here. SL.
 the computation of other
complicated Voronoi diagrams.  Our algorithm implies an $O(n \log^{3}
n)$ algorithm to compute the smallest disk touching or intersecting
all the input polygons.

We note that for a family of disjoint \emph{convex} polygons,
finding the smallest disk touching all of them is much easier, and can
be solved in time $O(n)$, where $n$ is the total complexity of the
polygons~\cite{JMB96}.

In Section~\ref{sec:prelim}, we start with some preliminaries and give a
definition of farthest-polygon Voronoi diagrams. We prove, in
Section~\ref{sec:complexity}, some properties on the structure of these
diagrams and bound their complexity. In Section~\ref{sec:computing}, we
present an algorithm for computing such diagrams and conclude in
Section~\ref{sec:conclusion}.

\section{Preliminaries} % and definition of farthest-polygon Voronoi diagrams  }
\label{sec:prelim}

We consider a family $\SITES$ of $k$ pairwise-disjoint polygonal sites
of total complexity~$n$.  Here, a \emph{polygonal site} of complexity
$m$ is the union of $m$ line segments, whose relative interiors are pairwise
disjoint, but whose union is connected, see Fig.~\ref{fig:site}. (In
other words, the corners\footnote{We reserve the word \emph{vertex} for
vertices of the Voronoi diagram.} and edges of a polygonal site form a
one-dimensional connected simplicial complex in the plane.)  In
particular, the boundary of a simple polygon is a polygonal site.  For
a point $x\in \R^{2}$, the distance $d(x, P)$ between $x$ and a site
$P \in \SITES$ is the Euclidean  distance from $x$ to the closest point on~$P$.

%The \emph{features} of a site $P$ are its corners and edges.  
%We say that two objects (sets) \emph{touch} if their closures intersect and their relative interiors
%are disjoint. However, we say  that a  \emph{disk touches a feature} if 
%their closure intersect in exactly one point and, if the feature is an edge, then
%its supporting line is tangent to the disk.

A \emph{pocket} $\pocket$ of $P$ is a connected component of
$\conv(P)\setminus P$, where $\conv(P)$ is the convex hull of~$P$.  A
pocket $\pocket$ is \emph{bounded} if it coincides with a bounded connected
component of $\R^{2} \setminus P$, \emph{unbounded} otherwise.  

The \emph{features} of a site $P$ are its corners and edges.  
We say  that a  disk \emph{touches a corner} if the corner lies on its boundary.
A disk \emph{touches an edge} if the closed edge touches the disk in one point and if its supporting
line is tangent to the disk. 
%feature  either if the corner lies on the boundary of the disk,
%or if the (closed) edge intersect the disk in  one point and its supporting line is tangent
%to the disk. 
A disk  \emph{touches a site} $P$ if the disk touches some of $P$'s features   and if the
disk's interior does not intersect $P$.

\referee{Referee \#1: "A disk touches a site P if the disk touches some of P's features and if the
disk's interior does not intersect P" But it may intersect the convex hull of P, if P is a polygonal
site?}
\respond{True, but we believe this is too obvious to be emphasized.}

\begin{figure}
  \centering
  \includegraphics{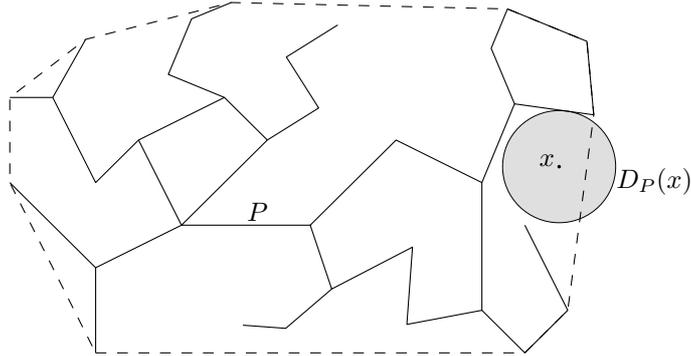}
  \caption{A polygonal site $P$ with three bounded pockets and seven unbounded ones, and its convex hull.
  A point $x$ and the disk $D_P(x)$.} 
%A polygonal site with three bounded pockets and nine unbounded ones.} 
  \label{fig:site}
\end{figure}

We assume that the family $\SITES$ is in general position, that is, no
disk touches four features, no line contains three corners, and no two
edges are parallel. 

Before defining the farthest-polygon Voronoi diagram of a family of polygonal sites, 
we define the medial axis of a polygonal site.
\paragraph{Medial axes.}
%The \emph{features} of a site $P$ are its corners and edges.  
For a site $P \in \SITES$, we define
the function $\DIST_{P}: \R^{2} \mapsto \R$ as $\DIST_{P}(x) = d(x,P)$. 
\referee{Referee \#2: Since functions $\Psi_P$ and $\Phi$ do not take on negative values, you might use
the subscript $\geq 0$ to denote their image spaces.}
\respond{This is correct but we believe that the benefit of specifying it is not worth introducing a new
notation $\R^{+}$ or $\R^{\geq}$ which would need to be defined.}
The graph of $\DIST_{P}$ is
a \emph{Voronoi surface}; it is the lower envelope of circular cones for each corner of $P$ and of
rectangular wedges for each edge of~$P$.  The orthogonal projection of this surface on the plane
induces a subdivision of the plane:
%, called the \emph{medial axis subdivision}: 
each 2D cell of this subdivision corresponds to a feature $w$ of $P$, and it is the set of all points $x\in
\R^{2}$ such that $w$ is or contains the unique closest point on $P$ to~$x$ (here, edges of $P$ are
considered relatively open, so the cell of a corner is disjoint from the cells of its incident
edges).  
The  \emph{medial axis} of $P$, denoted $\MA(P)$, consists of  the \emph{arcs} and \emph{vertices}
formed by  the boundaries between these cells. By extension, we call the cells of the subdivision the \emph{cells} of the medial axis.

%\footnote{Note that the medial axis is often defined as the locus of the centers of circles that are
%tangent to $P$ in two or more points. Our definition considers as well arcs that separate the
%cell of a reflex corner from the cell of an incident edge, that is the locus of centers of disks
%touching $P$ only at one vertex and tangentially to one of its incident edge.}

For a point $x\in \R^{2}$ and a site $P$, let $D_{P}(x)$ denote the
largest disk centered at $x$ whose interior does not intersect~$P$
(and which is therefore touching~$P$, see Fig.~\ref{fig:site}).  If $D_{P}(x)$ touches $P$ in a
single feature $w$, then $x$ lies in the cell of the medial axis subdivision associated with~$w$.  If $D_{P}(x)$
touches $P$ in two different features, then $x$ lies on an arc of
$\MA(P)$, and if $D_{P}(x)$ touches $P$ in three different features,
then $x$ is a vertex of~$\MA(P)$.  
Note that $\MA(P)$ contains some special arcs, called \emph{spokes}, that separate the
cell of a corner from the cell of an incident edge (see Fig.~\ref{fig:closed-curve}(b)); a spoke  arc is the locus of centers
of circles that intersect $P$ only at  a corner and that are tangent to the line supporting one of its incident edges.\footnote{Note that the medial axis of $P$ is often defined as the locus of the centers of circles that are
tangent to $P$ in two or more points. Such a definition would not
include spokes as part of the medial axis.}

%\SL{The previous sentence is modified accordingly with the def of touching.}

%\complain{A figure showing a possible medial axis
%with some spokes would be nice.  Added a ref to Fig1b.}

%The \emph{features} of a site $P$ are its corners and edges.  For a site $P \in \SITES$, we define
%the function $\DIST_{P}: \R^{2} \mapsto \R$ as $\DIST_{P}(x) = d(x,P)$. The graph of $\DIST_{P}$ is
%a \emph{Voronoi surface}, it is the lower envelope of circular cones for each corner of $P$ and of
%rectangular wedges for each edge of~$P$.  The orthogonal projection of this surface on the plane
%induces a subdivision of the plane, which we call the \emph{medial axis}~$\MA(P)$ of~$P$: Each cell
%of this subdivision corresponds to a feature $w$ of $P$, and it is the set of all points $x\in
%\R^{2}$ such that $w$ is or contains the unique closest point on $P$ to~$x$ (here, edges of $P$ are
%considered relatively open, so the cell of a corner is disjoint from the cells of its incident
%edges).  The boundaries between cells are the \emph{arcs} and \emph{vertices} of the medial axis.
%(Note that the medial axis is usually defined as the locus of the centers of circles that are
%tangent to $P$ in two or more points, which is a subset of these arcs and vertices; indeed, arcs
%induced by two features that are a reflex vertex and an incident edge correspond to circles tangent
%to $P$ in only one point.)

The medial axis $\MA(P)$  restricted to $\R^{2} \setminus P$ forms a forest.  
By this definition, the arc endpoints that lie on $P$ (at a corner) are not part of the forest; we
consider nonetheless these endpoints to be leaves of the forest. It follows that several leaves  may
coincide at a corner. More precisely, each convex
angle around a corner induces a leaf at that corner, and each reflex angle around a corner induces
two leaves incident to two spokes  at that corner. Spokes are always incident to a leaf  at a corner.  
Notice that the
leaves  always lie at  the corners
of~$P$ or at infinity.
%
%The medial axis $\MA(P)$  restricted to $\R^{2} \setminus P$ forms a forest.  The
%leaves of this forest lie at infinity or ``near'' the corners
%of~$P$. More precisely, a corner appears as a leaf for each convex
%angle at the corner, and as two leaves incident to two spokes for a
%reflex angle.  Spokes are always incident to leaves of the forest.  

The medial axis $\MA(P)$ consists of exactly one tree for each pocket
of~$P$, and two isolated spokes (with endpoints at infinity) for each
edge of $P$ appearing on $\conv(P)$.  The tree $\TREE$ of an unbounded
pocket $\pocket$ contains exactly one infinite arc, all other leaves
of $\TREE$ are corners of $\pocket$.  The leaves of the tree $\TREE$
of a bounded pocket $\pocket$ are exactly the corners of~$\pocket$.

\paragraph{Farthest-polygon Voronoi diagrams.}
We now consider the function $\UDIST: \R^{2} \mapsto \R$ defined as
$\UDIST(x) = \max_{P \in \SITES} \DIST_{P}(x)$.  The graph of $\UDIST$
is the upper envelope of the surfaces $\DIST_P$, for $P \in \SITES$.
The surface $\UDIST$ consists of conical and planar patches from the
Voronoi surfaces $\DIST_{P}$, and the arcs separating such patches are
either arcs of a Voronoi surface $\DIST_{P}$ (we call these
\emph{medial axis arcs}), or intersection curves of two Voronoi
surfaces $\DIST_{P}$ and $\DIST_{Q}$ (we call these \emph{pure arcs}).
\referee{Referee \#2: Please explain what types of curves occur (you are going to need that their
degree is bounded by 2).}
\respond{Done here and line 2/3 of the following paragraph for the 2D arcs.}
These arcs are hyperbolic arcs that lie in vertical planes, parabolic arcs, or straight-line segments.
They correspond respectively
to the intersection of two cones, a cone and a plane, and two planes.
The vertices of $\UDIST$ are of one of the following three types:
\begin{di} 
\item Vertices of one Voronoi surface $\DIST_{P}$. We 
  call these \emph{medial axis vertices}.
\item Intersections of an
  arc of $\DIST_{P}$ with a patch of another surface $\DIST_{Q}$. We
  call these \emph{mixed vertices}.
\item Intersections of patches of
  three Voronoi surfaces $\DIST_{P}, \DIST_{Q}, \DIST_{R}$. We call
  these \emph{pure vertices}.
\end{di}

The projection of the graph of  $\UDIST$ onto the plane induces 
%is %% The projection is the whole plane! SL.
the \emph{farthest-polygon Voronoi diagram} $\FS$ of $\SITES$. It is a
subdivision of the plane into cells, arcs, and vertices.  
The arcs are either parabolic or straight, since   hyperbolic arcs that lie in 
vertical planes project into line segments.
%We identify the pure, mixed, and medial
%axis arcs and vertices of $\UDIST$ with
%their orthogonal projection on the plane. 
Each cell
corresponds to a feature $w$ of a site $P$, the feature is the nearest
among the features of $P$, but is further away than the nearest
feature of any other site.  
The arcs and vertices of $\FS$ are the orthogonal projections of the arcs and vertices of $\UDIST$
and they inherit their types (pure, mixed, and medial
axis).
The farthest-polygon Voronoi diagram is
therefore completely analogous to the farthest-color Voronoi
diagram~\cite{ahiklmps-fcvdr-01}.

For a point $x\in \R^{2}$, let  $D(x) = D_{\SITES}(x)$ denote the
smallest disk centered at~$x$ that intersects all sites $P\in
\SITES$. By definition, there is always at least one site that touches
$D(x)$ without intersecting its interior, and the radius of $D(x)$ is
equal to~$\UDIST(x)$.  By our general position assumption, only the
following five cases can occur. 
%In the sequel, we say that a disk \emph{touches a feature} if 
%their closure intersect in exactly one point and, if the feature is an edge, then
%its supporting line is tangent to the disk.
\begin{di}
\item If $D(x)$ touches one site $P$ in only one feature $w$, and all
  other sites intersect the interior of $D(x)$, then $x$ lies in a
  cell of $\FS$, and the cell belongs to the feature $w$
  of~$P$.
\item If $D(x)$ touches one site $P$ in two or three features, and all
  other sites intersect the interior of $D(x)$, then $x$ lies on a
  medial axis arc or medial axis vertex of $\FS$, and is incident to
  cells belonging to different features of~$P$.
\item If $D(x)$ touches one feature $w$ of site $P$, one feature $u$
  of site $Q$, and all other sites intersect the interior of $D(x)$,
  then $x$ lies on a pure arc separating cells belonging to features
  $w$ and~$u$.
\item If $D(x)$ touches two features of site $P$ and one feature of
  site $Q$, and all other sites intersect the interior of $D(x)$, then
  $x$ is a mixed vertex incident to a medial axis arc of~$P$.
\item If $D(x)$ touches one feature each of three sites $P$, $Q$, and
  $R$, and all other sites intersect the interior of $D(x)$, then $x$
  is a pure vertex.
\end{di}
Put differently, vertices of $\FS$ are points $x\in \R^{2}$
where $D(x)$ touches three distinct features of sites.  If all three
features are on the same site, the vertex is a medial axis vertex.  If
the three features are on three distinct sites, then the vertex is a
pure vertex.  In the remaining case, if two features are on a site
$P$, and the third feature is on a different site~$Q$, the vertex is a
mixed vertex.

%For a point $x\in \R^{2}$ and a site $P$, let $D_{P}(x)$ denote the
%largest disk centered at $x$ whose interior does not intersect~$P$
%(and which is therefore touching~$P$).  If $D_{P}(x)$ touches $P$ in a
%single feature $w$, then $x$ lies in the cell of $\MA(P)$ associated with~$w$.  If $D_{P}(x)$
%touches $P$ in two different features, then $x$ lies on an arc of
%$\MA(P)$, and if $D_{P}(x)$ touches $P$ in three different features,
%then $x$ is a vertex of~$\MA(P)$.  
%%Note, however, that according to
%%our definition, $\MA(P)$ contains another kind of arc, separating the
%%cell of a corner from the cell of an incident edge.  We will call such
%%arcs \emph{spokes}.  
%Note that $\MA(P)$ contains some special arcs, called \emph{spokes}, that separate the
%cell of a (reflex) corner from the cell of an incident edge; such an  arc is the locus of the centers
%of circles that intersect $P$ only at  the corner and that are tangent to the line supporting the incident edge. 
%\SL{The previous sentence is modified accordingly with the def of touching.}

\begin{figure}[t]
  \centering
  \includegraphics{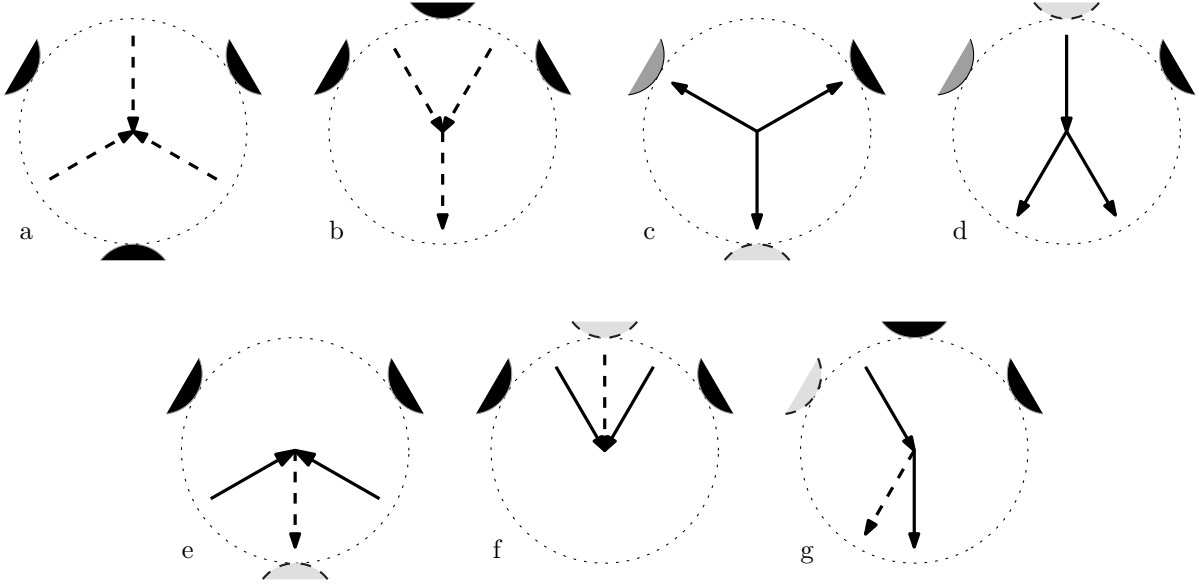}
  \caption{The different types of vertices in the farthest-polygon
    Voronoi diagram. Pure arcs 
    are shown solid, medial axis arcs dashed. The arrows indicate the
    direction of increasing~$\UDIST(x)$ in a neighborhood of the vertex.}
  \label{fig:vertices}
\end{figure}

Fig.~\ref{fig:vertices} illustrates all different vertex types.  Figs.~\ref{fig:vertices}(a) and \ref{fig:vertices}(b)
show the possible medial axis vertices; they 
differ in whether the triangle formed by the three features contains
the vertex or not. Similarly, Figs.~\ref{fig:vertices}(c) and \ref{fig:vertices}(d) show the pure vertex
types. 
%Fig.~\ref{fig:vertices} illustrates all different vertex types.  The
%top row shows the possible medial axis vertices: types~(a) and~(b)
%differ in whether the triangle formed by the three features contains
%the vertex or not.  Similarly, the middle row shows the pure vertex
%types, with the same distinction between~(d) and~(e).  
The bottom row
shows the three possible types of mixed vertices.  Again, we have to
distinguish whether the three features enclose the vertex or not, and
in the latter case we need to distinguish which two features are on
the same site.

Consider an arc $\alpha$ of $\FS$.  If a point  $x$ moves continuously along $\alpha$,
then $\UDIST(x)$---which is the radius of $D(x)$---changes
continuously.  The local shape of $\FS$ in a neighborhood of a vertex
$v$ is determined solely by the features defining the vertex.  For
each type of vertex shown in Fig.~\ref{fig:vertices}, we can therefore
uniquely determine whether $\UDIST(x)$ increases or decreases along an
arc in a \emph{neighborhood} of the vertex.  We indicate the direction of
increasing $\UDIST(x)$ along an arc by arrows in the figure.

%\begin{lemma}
%  \label{lem:local-maximum}
%  A local maxima of the restriction of $\UDIST(x)$ to the arcs of $\FS$ can only appear at vertices of $\FS$ of
%  type~(a) and~(f). 
%\end{lemma}
%\begin{proof}
%  Since an arc $\alpha$ is defined by two features (corners or edges), the graph of $\UDIST(x)$
%restricted to $\alpha$ is the intersection of the two Voronoi surfaces, which can be cones or
%wedges, induced by the two features. Thus $\UDIST(x)$ cannot have a local maximum in the interior
%of~$\alpha$ (it may have a local minimum).  This implies that the only local maxima of the
%restriction of $\UDIST(x)$ appear at vertices. A glance at the orientation of the arrows for each
%type of vertex in Fig.~\ref{fig:vertices} tells us that the only possible local maxima of the
%restriction of $\UDIST(x)$ are the vertices of type~(a) and~(f).
%\end{proof}

In addition to the seven types of vertices discussed above, we 
need to consider \emph{vertices at infinity}, that is, we consider the
semi-infinite arcs of $\FS$ to have a degree-one vertex at
their end.  For a vertex $v$ at infinity, the ``disk'' $D(v)$ is a
halfplane, and we have two cases:
\begin{di}
\item If $D(v)$ touches one site $P$ in two features, and all other
  sites intersect the interior of $D(v)$, then $D(v)$ is the
  ``infinite'' endpoint of a medial axis arc, and we consider it a
  medial axis vertex at infinity.
\item If $D(v)$ touches two distinct sites, and all other sites
  intersect its interior, then $D(v)$ is the ``infinite'' endpoint of
  a pure edge, and we consider it a pure vertex at infinity.
\end{di}

Finally, we define the \emph{Voronoi region} of a site $P \in
\SITES$.  The Voronoi region $\REG(P)$ of $P$ is simply the union of
all cells, medial axis arcs, and medial axis vertices of $\FS$
belonging to features of~$P$.  Voronoi regions are not necessarily
connected, as we will see in the next section.  We call each
connected component of a Voronoi region a \emph{Voronoi component}.

\section{Structure and complexity}
\label{sec:complexity}

%Consider the medial axis $\MA(P)$ of a polygonal site~$P$.  The arcs
%of $\MA(P)$ restricted to $\R^{2} \setminus P$ form a forest.  The
%leaves of this forest lie at infinity and ``near'' the corners
%of~$P$. More precisely, a corner appears as a leaf for each convex
%angle at the corner, and as two leaves incident to two spokes for a
%reflex angle.  Spokes are always incident to leaves of the forest.  In
%the following, \emph{we  identify the medial axis $\MA(P)$ with its
%forest.}
%
%A \emph{pocket} $\pocket$ of $P$ is a connected component of
%$\conv(P)\setminus P$, where $\conv(P)$ is the convex hull of~$P$.  A
%pocket $\pocket$ is \emph{bounded} if it coincides with a bounded connected
%component of $\R^{2} \setminus P$, \emph{unbounded} otherwise.  
%
%The medial axis $\MA(P)$ consists of exactly one tree for each pocket
%of~$P$, and two isolated spokes (with endpoints at infinity) for each
%edge of $P$ appearing on $\conv(P)$.  The tree $\TREE$ of an unbounded
%pocket $\pocket$ contains exactly one infinite arc, all other leaves
%of $\TREE$ are corners of $\pocket$.  The leaves of the tree $\TREE$
%of a bounded pocket $\pocket$ are exactly the corners of~$\pocket$.

%\SL{Added some organizational comments and moved the first 3 paragraphs to the previous section}

In this section, we prove some properties on the structure of farthest-polygon
Voronoi diagrams of polygonal sites and we bound their complexity.
Farthest-polygon Voronoi diagrams contain three different types of vertices, as
defined in the previous  section. Note first that the number of medial axis
vertices is bounded by the total complexity of all medial axes, which is
$O(n)$.
%linear in the size of the input polygonal sites.

Now, we first bound the number of mixed vertices. For that purpose, we show
that when a tree $\TREE$ of $\MA(P)$ intersects the Voronoi region~$\REG(P)$ of
$P \in \SITES$, then the intersection $\TREE \cap \REG(P)$ is a connected
subtree. We start with a preliminary lemma. 

%We first show that when a tree $\TREE$ of $\MA(P)$ intersects the Voronoi
%region~$\REG(P)$ of $P \in \SITES$, then the intersection $\TREE \cap
%\REG(P)$ is a connected subtree.  This will help bounding the number of mixed vertices. We start by
%a preliminary lemma. 
%%Note that a point $x\in \MA(P)$ is
%%in $\REG(P)$ if and only if $D_{\SITES}(x) = D_{P}(x)$, which is true
%%if and only if all other sites intersect the interior of~$D_{P}(x)$.
%\SL{I don't understand the point of the  sentence (which I removed)  \emph{Note that a point $x\in \MA(P)$ is
%in $\REG(P)$ if and only if $D_{\SITES}(x) = D_{P}(x)$, which is true
%if and only if all other sites intersect the interior of~$D_{P}(x)$.} (Moreover,
%it is true for all x in the plane.)}
\begin{lemma}
  \label{lem:path}
  Let $\gamma$ be a path in $\MA(P)$, let $Q \in \SITES\setminus\{P\}$
  be another site, and let $\gamma_{Q}$ be the 
%set of points $x \in  \gamma$ where $D_{P}(x)$ intersects~$Q$.  
part of $\gamma$ that is closer to $Q$ than to~$P$.
%\SL{Modified as suggested by Joachim (and modified accordingly the proof of the next lemma).}
Then $\gamma_{Q}$ is a
  connected subset of $\gamma$, that is, a subpath.
\end{lemma}
\begin{proof}
We can assume $\gamma$ to be a maximal path in $\TREE$, connecting a corner $w$
of $P$ with another corner~$u$ or a medial-axis vertex at infinity. Assume for
a contradiction that there are points $x$, $y$, $z$ on $\gamma$ in this order
such that $x, z \in \gamma_{Q}$, but $y\not\in \gamma_{Q}$.

%\SL{I did not find the rest of the proof very clear and I suggest the following new version.}

We first consider the case where $y$ lies on a spoke of $\MA(P)$ induced by a
corner $c$ of $P$ (and one of its incident edges). The  spoke is incident to a
leaf of $\MA(P)$ located at $c$ and, without loss of generality, $x$ lies on
this spoke between $c$ and $y$. Since $y\not\in \gamma_{Q}$, the disk $D_P(y)$
centered at $y$ and touching $c$ does not intersect $Q$. The disk $D_P(x)$ is
included in $D_P(y)$ and thus, it  does not intersect $Q$ either, which
contradicts our hypothesis. 

\newcommand{\subspace}{\ensuremath{{\Omega}}\xspace}

\referee{Referee \#1: Proof of Lemma 1, last paragraph: 'The disk $D_p(x)$ touches $P$ on the boundary
of A...' that depends on your choice of labels A,x, etc. 'Hence, since Q is connected, Q
intersects...' The important part here is that you assume P and Q to be disjoint, not that Q is
connected.}
\referee{Referee \#2: You need to argue what happens if both x and z are contained in A cup $D_P(y)$
but not in B.}
\respond{We have rewritten the end of the proof to clarify the proof and these issues.}

We now consider the case where $y$ does not lie on a spoke of $\MA(P)$. 
Let \subspace be the connected component of $\R^{2}\setminus P$ that contains $y$. Since $y$
does not lie on a spoke of $\MA(P)$, the disk $D_{P}(y)$ touches $P$ in $k\geq 2$ distinct points
($k\leq 3$ by the general position assumption), and thus
$D_{P}(y)$ partitions \subspace into $k+1$ connected components: $D_{P}(y)$ and $k$ other components
denoted $A_1,\ldots, A_k$. Since $P$ is a polygon, the structure of its medial axis is well understood. 
In particular, in the neighborhood of $y$, $\MA(P)$ consists of $k$
arcs (straight or parabolic), and for every point $p$ on any single one of these arcs, $D_{P}(p)$
intersects one and the same  components $A_i$, and is contained in $A_i\cup D_{P}(y)$.
Also, since $\MA(P)$ consists of $k$ arcs in the neighborhood of $y$,   point $y$ splits the medial axis tree $\TREE$ that contains $\gamma$ into $k$
subtrees  $\TREE_1,\ldots,\TREE_k$.

Now, we observe that any open disk $D$, that does not intersect $P$, cannot
contain points in two distinct  components $A_i$ and $A_j$. Indeed, the boundary of $D$
would have to intersect the boundary of $D_{P}(y)$ in at least four points, implying that the two
disk coincide, and thus that $D$ intersects neither $A_i$ nor $A_j$ (since $D$ is open).

For any $p\in \TREE_i$, distinct from $y$, $D_{P}(p)$ is  included in $\subspace$
but  not in $D_{P}(y)$. Thus, the  interior of $D_{P}(p)$  intersects $\bigcup_j A_j$ but not $P$. Hence, it % the interior of $D_{P}(p)$
intersects only one of the  $A_j$. Furthermore, by continuity, the interior of $D_{P}(p)$
intersects the same $A_j$ for all $p$ in $ \TREE_i$. We can thus assume without loss of
generality that, for all $p$ in $\TREE_i$, the interior of $D_{P}(p)$
intersects $A_i$ and none of the other $A_j$. Since $D_{P}(p)$ lies in
$\subspace$, it  also follows that, for all $p$ in $\TREE_i$, $D_{P}(p)$ lies in $A_i\cup D_{P}(y)$.

Now, $x$ and $z$ belong to two distinct subtrees, say $\TREE_1$ and $\TREE_2$, respectively.
Thus, $D_P(x)$ lies in $A_1\cup D_{P}(y)$ and $D_P(z)$ lies in $A_2\cup D_{P}(y)$.
By assumption, both $D_{P}(x)$ and $D_{P}(z)$ intersect $Q$, thus $Q$ intersects
  both $A_1\cup D_{P}(y)$ and $A_2\cup D_{P}(y)$. Hence, since $Q$  is connected and does not intersect $P$,
  $Q$ must intersect $D_{P}(y)$, which is a contradiction and concludes the proof.
%
%========TALG:
%
%We now consider the case where $y$ does not lie on a spoke of $\MA(P)$. Then,
%the disk $D_{P}(y)$ touches $P$ in two distinct points, and 
%%  The disk
% $D_{P}(y)$ partitions one connected component of $\R^{2}\setminus P$ into
% three components $A$, $D_{P}(y)$, and $B$.  The two points $x$ and $z$ of
% $\gamma$ must lie, respectively, in, say, $A\cup D_{P}(y)$ and $D_{P}(y)\cup
% B$.
%
%  We first argue that any disk $D$ that does not intersect $P$ cannot
%  contain points in both $A$ and $B$. Indeed, the boundary of $D$
%  would have to intersect the boundary of $D_{P}(y)$ in four points, a
%  contradiction.
%
%  The disk $D_{P}(x)$ touches $P$ on the boundary of $A$, and so it contains a
%  point in $A$, which implies that it cannot contain a point in $B$;  hence
%  $D_{P}(x)$ lies in $A\cup D_{P}(y)$, since it lies in the component $A\cup
%  D_{P}(y)\cup B$.  Similarly, $D_{P}(z)$ lies in $D_{P}(y)\cup B$. By
%  assumption, both $D_{P}(x)$ and $D_{P}(z)$ intersect $Q$, thus $Q$ intersects
%  both $A\cup D_{P}(y)$ and $D_{P}(y)\cup B$. Hence, since $Q$  is connected,
%  $Q$ intersects $D_{P}(y)$, which is a contradiction and concludes the proof.
\end{proof}
\begin{lemma}
  \label{lem:tree}
  Let $\TREE$ be a tree of $\MA(P)$.  Then $\TREE \cap \REG(P)$ is a
  connected subtree of $\TREE$.
\end{lemma}
\referee{Referee \#1: Proof of Lemma 2: It took me a while to understand this one. Maybe a proof by
contradiction would be easier to grasp?}
\respond{We have not modified this proof as we believe that it is quite clear and that it would be
difficult to make it clearer.} 
\begin{proof}
  Pick two points $p, q \in \TREE \cap \REG(P)$, let $\gamma$ be the
  path on $\TREE$ from $p$ to $q$, and let $x$ be any point between
  $p$ and $q$ on~$\gamma$.  
%We need to show that $x \in \REG(P)$,
%  which is equivalent to showing that every site $Q\neq P$ intersects
%  the interior of $D_{P}(x)$.  Let $Q$ be such a site.  Since $p, q
%  \in \REG(P)$, $Q$ intersects the interior of $D_{P}(p)$ and
%  $D_{P}(q)$. By Lemma~\ref{lem:path}, this implies that $Q$ intersects
%  the interior of $D_{P}(x)$.
We need to show that $x \in \REG(P)$,
  which is equivalent to showing that $x$ is closer to every site
  $Q\neq P$ 
  than to~$P$.  Let $Q$ be such a site.  Since $p, q
  \in \REG(P)$, $p$ and $q$ are closer to $Q$ than to~$P$. 
  By Lemma~\ref{lem:path}, this implies that $x$ is closer to  $Q$ than to~$P$.
\end{proof}
The lemma allows us to bound the number of mixed vertices of~$\FS$.
\begin{lemma}
  \label{lem:number-mixed}
  The number of mixed vertices of $\FS$ for a family of disjoint
  polygonal sites of total complexity~$n$ is $O(n)$.
\end{lemma}
\begin{proof}
  Consider a site $P \in \SITES$ of complexity~$m$.  Its medial axis
  $\MA(P)$ has complexity~$O(m)$. By Lemma~\ref{lem:tree}, for
  each tree~$\TREE$ of~$\MA(P)$ the intersection $\TREE \cap \REG(P)$
  is a connected subtree.  Since the mixed vertices on~$\TREE$ are
  exactly the finite leaves of this subtree, this implies that the
  number of mixed vertices on $\MA(P)$ is~$O(m)$. Summing over all $P
  \in \SITES$ then proves that the number of mixed vertices of $\FS$
  is~$O(n)$.
\end{proof}
We next consider the vertices at infinity.
\begin{lemma}
  \label{lem:linear-at-infinity}
  The number of pure vertices at infinity of $\FS$ is at most~$2k-2$.
  The total number of vertices at infinity of $\FS$ is~$O(n)$.
\end{lemma}
\begin{proof}
  For two sites $P, Q \in \SITES$, consider the diagram~$\fpvd(\{P,
  Q\})$.  A pure vertex at infinity corresponds to an edge of $\conv(P
  \cup Q)$ supported by a corner of~$P$ and a corner of~$Q$.  But
  $\conv(P \cup Q)$ can have at most two such edges, since $P$ and $Q$
  are disjoint and both are connected, and so $\fpvd(\{P, Q\})$ has at
  most two pure vertices at infinity.

  Consider now again $\FS$, and let $\sigma(\SITES)$ denote the
  sequence of sites whose Voronoi regions appear at infinity in
  circular order, starting and ending at the same region. We claim
  that $\sigma(\SITES)$ is a Davenport-Schinzel sequence of order~2,
  and has therefore length at most $2k-1$~\cite{sa-dsstg-95}.  Indeed,
  $\sigma(\SITES)$ has by definition no two consecutive identical
  symbols.  Assume now that there are two sites $P$ and $Q$ such that
  the subsequence $PQPQ$ appears in $\sigma(\SITES)$. If we delete all
  other sites, then $\sigma(\{P, Q\})$ would still need to contain the
  subsequence $PQPQ$, and therefore $\fpvd(\{P,Q\})$ would contain at
  least three pure vertices at infinity, a contradiction to the
  observation above.

  It now suffices to observe that the pure vertices at infinity are
  exactly the transitions between consecutive Voronoi regions, and
  their number is at most~$2k-2$. 

  All remaining vertices at infinity are medial axis vertices. 
Since the total complexity of all medial axes is $O(n)$, the bound follows
\end{proof}

We proved so far that the number of mixed and medial axis vertices is  $O(n)$  and, furthermore,
that there are at most $2k-2$  pure vertices at infinity. 
It remains to  bound the other pure vertices, for which  we first
need to prove a few basic properties. 

We start by  discussing
 a monotonicity property of cells of~$\FS$.  Let $C$ be
a cell of $\FS$ belonging to feature $w$ of site~$P$.  For a point $x
\in C$, let $x\s$ be the point on $w$ closest to~$x$.  Let $f_{x}$ be
the %a 
directed line segment starting at~$x$ and extending in
direction~$\overrightarrow{x\s x}$ until 
it reaches %we reach
 $\MA(P)$ (a semi-infinite segment if this does not happen).  
%\SL{Changed "a directed line segment" into "the ..." and "we reach" into "it reaches"}
We call $f_{x}$ the
\emph{fiber} of~$x$.  We note that if $w$ is an edge, then all fibers
of $C$ are parallel, and normal to $w$; if $w$ is a corner then all
fibers are supported by lines through~$w$.
\referee{Referee \#1: 'We note that if w is an edge, then all fibers of C are parallel and normal to
w'. I did not see that all until I realized that the distance function for farthest VD is from x to
the closest point on each site P (as stated on page 3). It might be worth to remind the reader of
this fact?}
\respond{We do not quite see the process of mind of the referee here and we have not made any change.}
\begin{lemma}
  \label{lem:fiber}
  For any $x \in C$, the fiber $f_{x}$ lies entirely in $C$ (and
  therefore in $\REG(P)$).
\end{lemma}
\begin{proof}
  The disk $D(x)$ touches $P$ in $x\s$ only, and its interior
  intersects all other sites.  When we move a point $y$ from $x$ along
  $f_{x}$, the disk $D$ centered at $y$ through $x\s$ keeps containing
  $D(x)$, and it therefore still intersects all other sites.  This
  implies that $y \in C$ as long as $D$ does not intersect $P$ in
  another point.  This does not happen until we reach $\MA(P)$.
\end{proof}
An immediate consequence, which we will use for computing  Voronoi diagrams
(Section~\ref{sec:computing}),  is that cells are ``monotone'':
%\begin{lemma}
%  \label{lem:monotonicity}
%  Let $C$ be a cell of $\FS$ belonging to feature $w$.  If $w$ is a corner,
%  then any line through $w$ intersects $C$ in a segment.  If $w$ is an
%  edge, then any line normal to $w$ intersects $C$ in a segment.
%\end{lemma}
%\begin{proof}
%  Consider such a line $\ell$, and let $x$ be the point closest to~$w$
%  in~$\ell\cap C$.  Then the entire fiber $f_{x}$ lies in $C$, and no
%  point on $\ell$ beyond the medial axis can be in $C$.
%\end{proof}
%Lemma~\ref{lem:monotonicity} implies that the boundary of a cell~$C$
%belonging to a feature $w$ consists of two chains monotone with
%respect to~$w$ (that is, monotone in the direction of an edge, and
%rotationally monotone around a corner).  The \emph{lower chain} is
%closer to the feature and consists of pure arcs only, the \emph{upper
%  chain} consists of medial axis arcs only. 
%
\begin{lemma}
  \label{lem:monotonicity2}
%  Let $C$ be a cell of $\FS$ belonging to feature $w$.  If $w$ is a corner,
%  then any line through $w$ intersects $C$ in a segment.  If $w$ is an
%  edge, then any line normal to $w$ intersects $C$ in a segment. 
%
The boundary of a cell~$C$ of $\FS$ belonging to feature $w$ consists of two
chains monotone with respect to~$w$, that is, monotone in the direction of $w$
if $w$ is an edge, and rotationally monotone around $w$ if $w$ is  a corner.
The \emph{lower chain} is closer to the feature and consists of pure arcs only,
the \emph{upper chain} consists of medial axis arcs only.  
\end{lemma}
\referee{Referee \#1: Proof of Lemma 6: I am missing a statement to the fact the lower chain consists of pure arcs only.}
\referee{Referee \#2: Why is it not possible to re-enter cell C when walking beyond the first point
of the medial axis (couldn't the medial axis cross line l twice)? Or what do you mean by "monotone"?
Also, it would make more sense to call the chains "near" and "far" instead of "lower" and "upper".}
\respond{We partially rewrote this proof to address these relevant issues. We however kept the name
of lower and upper chain as it is a usual name and it carries the correct meaning.}
\begin{proof}
Let $P$ be the site containing the feature $w$.
Consider a half-line~$\ell$ with origin %$o$ 
on~$w$, and normal to~$w$ if $w$ is an edge.
%, and through $w$if $w$ is a corner.
Let $x$ be the point closest to~$w$ in~$\ell\cap C$.
It is straightforward that  the
entire fiber $f_{x}$ lies in $C$, and no point $z$ on $\ell$ beyond the medial axis
can be in $C$ since the feature of $P$ closest to any such $z$ cannot be $w$.
Hence, the boundary of $C$ consists of two monotone chains with respect to~$w$. 
Moreover, the upper chain consists of medial axis arcs by definition of the fibers $f_{x}$.
%\SL{Sam, I reread both your version and mine after having forgotten them both and I find mine 
%simpler and shorter (I actually simplified it even more). What do you think?} 
The lower chain consists of pure arcs, because if a point $x$ on the lower chain was on the
medial axis, then the fiber $f_{x}$ would be reduced to point $x$, by definition; thus $x$ would
also be on the upper chain, implying that $x$ is an endpoint of the two chains.
%
%% SAM's TRY:
%Consider the disk $D_P(y)$ for a point $y$ traveling on the straigth-line
%segment from some point $x'\in f_x$ down to $x$ on the lower chain.  $D_P(y)$
%stays tangent to $w$ at $x^*$ and its radius decreases. Thus $y$ can not meet the medial
%axis of $P$ and must reach a pure arc at $x$ when one site $Q\not=P$ no longer
%intersects the interior of $D_P(y)$. Therefore the lower chain consists of
%pure arcs only.
%% END OF SAM's TRY.
% SYLVAIN'S TRY:
%The lower chain consists of pure arcs, because if a non-terminal point $x$ on the lower chain was on the
%medial axis, then the fiber $f_{x}$ would be reduced to point $x$, by definition; thus $x$ would
%also be on the upper chain, implying that $x$ is an endpoint of the two chains, a contradiction. 
% END OF SYLVAIN'S TRY.
%
%Then $x$ must lie on a pure arc and the
%entire fiber $f_{x}$ lies in $C$. Furthermore, no point on $\ell$ beyond the medial axis
%can be in $C$ since 
%%no such point can have point $o$ as its closest point with respect to $P$.
%since the  feature of $P$ the closest to any such point cannot be $w$.
%\SL{Samuel, what do you think of this last sentence? If clear enough, it's lighter (and allows to
%avoid introducing $o$).}
\end{proof}

We now show (Lemma~\ref{lem:bounded-connected}) that if a Voronoi
region is bounded, then it is connected (we actually show that it is
simply connected, but will not use that fact in this paper). This
property is tight in the sense that, as shown in
Lemma~\ref{lem:unbounded-k-1}, a single Voronoi region may consist of
up to $k-1$ unbounded connected components; we postpone the proof of
this property to the end of the section.

\begin{lemma}
  \label{lem:contain}
  If a  connected component of the Voronoi region $\REG(P)$ of a site $P \in \SITES$
  is  bounded, then $P$ properly contains another site inside one   of its pockets.
\end{lemma}
\begin{proof}
Let $C$ be a bounded connected component of $\REG(P)$. 
  We first observe that $C$ contains some points of the medial
  axis $\MA(P)$ of $P$. Indeed, let $x\in C$ and consider its
  fiber $f_x$. By Lemma~\ref{lem:fiber} and since $C$ is bounded, $f_x$ does not extend to
  infinity and, therefore, one of its endpoints  lies on $\MA(P)$.

  Let $x$ be a point in $C\cap\MA(P)$. If $x$ lies in a pocket
  of $P$ that does not share an edge with $\conv(P)$ (that is, the
  pocket is a hole in $P$), then this pocket does contain all the
  other sites in $\SITES\setminus\{P\}$ and the lemma is proven.
  Otherwise, we let $x$ move along the medial axis $\MA(P)$ up to
  infinity. At some point $x'$, the point must exit from the bounded region $C$. This
  means that $D_\SITES(x')=D_P(x')$ is tangent to another site $Q$
\referee{Referee \#2: Instead of "...site Q without intersecting it properly." one could say
"...site Q, but does no longer intersect it properly."}
\respond{Done.}
  and does no longer intersect it properly. It then follows from the fact that $x'$ lies on the medial axis
  of $P$ that 
the   site~$Q$ lies entirely in the pocket of $P$ associated with the tree of $\MA(P)$ containing
$x$ and $x'$. 
\end{proof}
%TEST

\begin{lemma}
  \label{lem:bounded-connected}
If a connected component of the Voronoi region $\REG(P)$ of a site $P \in \SITES$ is bounded, then $\REG(P)$ is simply
  connected.
\end{lemma}
\begin{proof}
By Lemma~\ref{lem:contain}, $P$ contains another site $Q$ in one of its pockets $\mathfrak{P}$. 
Let $\TREE$ be the tree of $\MA(P)$ that corresponds to $\mathfrak{P}$.
%, {\it i.e.}, the tree $\TREE$ such that for any point $x\in \TREE$, the disk $D(x)$ properly
%intersects the pocket   $\mathfrak{P}$. 

%%  We can first  assume that $\REG(P)$ is non-empty since, otherwise, it is simply connected.  Then,
% Let $C$ be a bounded connected component of $\REG(P)$. As
%  argued in the proof of Lemma~\ref{lem:contain}, there exists a point $p$ in $C\cap \MA(P)$,
%  and  the tree $\TREE$ of $\MA(P)$ that contains $p$ corresponds to a pocket $\mathfrak{P}$ of
%$P$ which properly contains another site $Q$.

Let $x$ be any point in $\REG(P)$. The disk $D(x)$ touches $P$ in a
point $x\s$, and its interior intersects all other sites,
including~$Q$. Hence, $D(x)$ properly intersects the
pocket~$\mathfrak{P}$. It follows that when moving a point $y$ from
$x$ in the direction of $\overrightarrow{x\s x}$, the disk centered at
$y$ through $x\s$ keeps containing $D(x)$, and it therefore keeps
intersecting $\mathfrak{P}$.  
\referee{Referee \#2: What if fiber $f_x$ is unbounded so that no finite $y=x'$ exists?}
\respond{This case is not possible. We emphasized the fact that the point $y$ is necessarily finite.}
This implies that, at some finite point $y=x'$,
the disk $D$ centered at $x'$ and tangent to $P$ at $x\s$ becomes
tangent to $P$ at some other point, hence $x'$ lies on
$\MA(P)$. Moreover, $x'$ lies on the tree $\TREE$ because the disk $D$
properly intersects $\mathfrak{P}$.

Hence, for any point $x$ in $\REG(P)$, the fiber $f_x$ is a segment joining $x$ to a point $x'$ on
$\TREE$. Furthermore, the fiber $f_x$ lies in $\REG(P)$, by Lemma~\ref{lem:fiber}, and $\TREE \cap \REG(P)$
is a connected tree, by Lemma~\ref{lem:tree}.
Therefore, $\REG(P)$ is connected. 

It remains to show that $\REG(P)$ is simply connected.
We have shown that for any $x\in\REG(P)$, the
fiber $f_x$  is  a segment contained in $\REG(P)$, and 
connecting $x$ to $x'\in\TREE\cap \REG(P)$. 
By moving the points of $\REG(P)$ along their fiber, we can design, as follows, a continuous
deformation retraction of $\REG(P)$ onto the tree $\TREE\cap\REG(P)$ which implies that $\REG(P)$
is simply connected.

More precisely, it easy to check that the map $F:\REG(P)\times[0,1]\rightarrow\mathbb{R}^2$, $(x,t)\mapsto (1-t)\,x+t\,x'$
is continuous. Furthermore, $F$ is a (strong) deformation retraction since, for all $x\in\REG(P)$,
$t\in[0,1]$ and  $m\in\TREE\cap \REG(P)$, we have $F(x,0)=x$, $F(x,1)\in\TREE\cap \REG(P)$, and $F(m,t)=m$.
We have thus exhibited a deformation retraction of $\REG(P)$ onto $\TREE\cap\REG(P)$, which implies that these two point sets have the same fundamental group~\cite[Proposition 1.17]{hatcher}. Since $\TREE\cap\REG(P)$ is a tree, 
% (by Lemma~\ref{lem:tree}), 
its  fundamental group is trivial and so is the fundamental group of $\REG(P)$. Any pair of points $x$
 and $y$ in $\REG(P)$ can be connected with a path made of their fibers $f_x$ and $f_y$ and a path
 in $\TREE\cap\REG(P)$, thus $\REG(P)$ is path connected. Together with having a trivial fundamental
 group, this property makes $\REG(P)$ simply connected~\cite[p.~28]{hatcher}.
\end{proof}

We can now conclude our analysis of the complexity of farthest-polygon Voronoi diagrams.
%\SL{After moving the lower bound lemma to the end of the section, it seems better to move the proof
%of the theorem inside a "proof"}

%We summarize our complexity results:
\begin{theorem}
  \label{thm:size}
%  The farthest-polygon Voronoi diagram of a family of $k$~disjoint
%  polygonal sites of total complexity~$n$ consists of at most~$2k-2$
%  Voronoi components, $O(k)$ pure vertices, and has total
%  complexity~$O(n)$.
%
  The farthest-polygon Voronoi diagram of a family of $k$~disjoint
  polygonal sites of total complexity~$n$ has $O(k)$ pure vertices and total
  complexity~$O(n)$. It consists of at most~$2k-2$
  Voronoi components and this bound is tight in the worst case.
\end{theorem}
\begin{proof}
The farthest-polygon Voronoi diagram contains three different kinds of
vertices.  The number of medial axis vertices is clearly only $O(n)$,
since the total complexity of all $\MA(P)$ for $P \in \SITES$ is
only~$O(n)$.  In Lemma~\ref{lem:number-mixed}, we showed that the
number of mixed vertices is also only~$O(n)$.  It remains to bound the
number of pure vertices of~$\FS$.

Let $k_1$ be the number of bounded Voronoi components. By Lemma~\ref{lem:bounded-connected}, each of
these components corresponds to a different site and only the remaining $k-k_1$ sites can
contribute to form vertices at infinity. By the proof of 
Lemma~\ref{lem:linear-at-infinity}, there are at most $2(k-k_1)-2$ pure
vertices at infinity, and therefore at most $2(k - k_1) - 2$ unbounded Voronoi
components.  It follows that the total number of Voronoi components is at most
$2k - k_1 - 2  \leq 2k-2$. Moreover, the construction of Lemma~\ref{lem:unbounded-k-1} shows that this bound is tight.

%Let $k_1$ be the number of sites whose Voronoi region is bounded (and therefore
%is empty or consists of a single component by
%Lemma~\ref{lem:bounded-connected}). Only the remaining $k-k_1$~sites can
%contribute to form vertices at infinity.  By the proof of 
%Lemma~\ref{lem:linear-at-infinity}, there are at most $2(k-k_1)-2$ pure
%vertices at infinity, and therefore at most $2k - 2k_1 - 2$ unbounded Voronoi
%components.  It follows that the total number of Voronoi components is at most
%$2k - 2k_1 - 2 + k_1 \leq 2k-2$.

Let us now consider the graph~$G$ formed by the pure arcs and pure
vertices of~$\FS$.  Mixed vertices appear as vertices of degree two
in~$G$ (see Fig.~\ref{fig:vertices}), medial axis vertices do not
appear at all.  The faces of~$G$ are exactly the Voronoi components.
Since $G$ has at most $2k-2$~faces, Euler's formula implies that $G$
has $O(k)$ vertices of degree three.
\end{proof}

Finally, we prove that, as mentioned above, a single
Voronoi region of~$\FS$ can have up to $k-1$ connected components.
%The total number of unbounded Voronoi components  being at most $2k-2$ by Theorem~\ref{thm:size}, one
%Voronoi region has at most $k-1$ unbounded components. 
%%%and thus at most $k-1$ components by .Lemma~\ref{lem:bounded-connected}
%%The total number of Voronoi components at infinity being at most $2k-2$, a
%%Voronoi region must also have at most $k-1$ components. 
%Hence, the $(k-1)$ bound is tight.
%
%\SL{This says that a
%Voronoi region must also have at most $k-1$ unbounded components. We have not (yet) proved that if a
%Voronoi region contains a bounded connected component then the region is connected. (Lemma 8 only shows
%that if a Voronoi region is bounded then it is connected.) The conclusion that the bound is tight is
%thus currently hasty.} 

%Voronoi regions in \emph{abstract farthest-site Voronoi diagrams}1234
%consist of at most two components.  
\begin{figure*}
%  \centerline{\includegraphics[width=8cm]{fig/n_cells}}
  %\centerline{\includegraphics[width=15cm]{fig/k-cells}}
  \centerline{\includegraphics{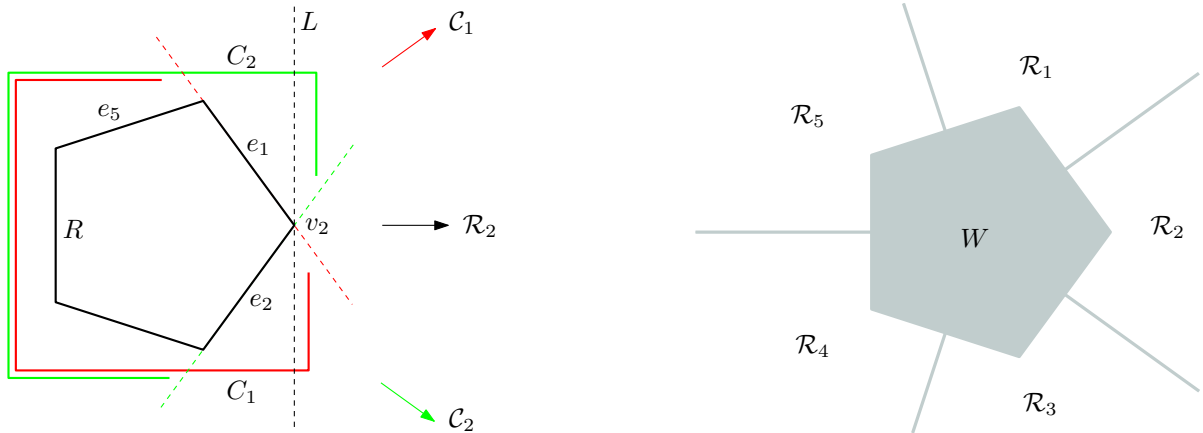}}
  \caption{The Voronoi region of polygon $R$
% the middle polygon 
has $k-1$     connected components.}
  \label{fig:n_cells}
\end{figure*}

\referee{Referee \#2: It would be helpful to add line L and labels $v_2$, L to the left part of Figure
4.}
\respond{Done}

\begin{lemma}
  \label{lem:unbounded-k-1}
  A single Voronoi region of~$\FS$ can have $k-1$ connected components and this bound is tight. 
\end{lemma}
\begin{proof}
  The construction is shown in Fig.~\ref{fig:n_cells}. It consists of
  one $(k-1)$-regular polygon $R$ and $k-1$ polygonal chains $C_1,
  \ldots, C_{k-1}$. Let $e_1, e_2, \ldots, e_{k-1}$ denote the edges
  of $R$ in circular order. We inductively construct the polygonal
  sites $C_i, i=1,2,\ldots,k-1$ as follows. For the supporting line $l$
  of $e_i$, let $l^{+}$ be the closed halfplane containing $R$ and
  $l^-$ be the other. Then, consider the intersection $C_i^*$ between
  $l^+$ and the four edges of a square that contains $R, C_1, \ldots,
  C_{i-1}$ inside. We define $C_i$ as the set of points of $C_i^*$
  whose distance to $l$ is larger than some fixed small
  $\varepsilon>0$.  $C_i$ has at most four edges. Note that $l^+$
  contains $C_i$ completely. Consider a ray from $e_i$ to infinity in
  $l^-$ which is orthogonal to $l$. Since $l^-$ intersects all sites
  but $C_i$, the endpoint at infinity of this ray lies in the region
  $\mathcal C_i=\REG(C_i)$. % at infinity.  
On the other hand, for a
  sufficiently small $\varepsilon$, there is a line $L$ passing
  through $v_i$ (the vertex incident to $e_{i-1}$ and $e_{i}$) such
  that we can define $L^+$ as an open halfplane containing
  $R\backslash \{v_i\}$ and $L^-$ as the other open halfplane
  intersecting all the other sites but $R$.  The  endpoint at infinity  of
  the ray from $v_i$ to infinity in $L^-$ which is orthogonal to $L$
  lies in a connected component of $\REG(R)$, which we call
  $\mathcal{R}_i$.
%\SL{There was a confusion of notation with line $l$ used with 2 meanings ($l$ and $L$).}

%On the other hand, for a
%  sufficiently small $\varepsilon$, there is a line $l$ passing
%  through $v_i$ (the vertex incident to $e_{i-1}$ and $e_{i}$) such
%  that we can define $l^+$ as an open halfplane containing
%  $R\backslash \{v_i\}$ and $l^-$ as the other open halfplane
%  intersecting all the other sites but $R$.  The infinite endpoint of
%  the ray from $v_i$ to infinity in $l^-$ which is orthogonal to $l$
%  lies in a connected component of $\REG(R)$, which we call
%  $\mathcal{R}_i$.

\referee{Referee \#1: 
Proof of Lemma 10, lines 51-56: I don't like the use of $\omega$ here, it is too close in appearance
to the $w$ you use to indicate a feature of P.}
\respond{Changed everywhere for $\phi$}

  Therefore at infinity 
%$\mathcal{C}_1, \mathcal{R}_1, \mathcal{C}_2,  \cdots , \mathcal{C}_{k-1}, \mathcal{R}_{k-1}$ 
$\mathcal{R}_{1}, \mathcal{C}_1, \mathcal{R}_2, \mathcal{C}_2,  \cdots , \mathcal{R}_{k-1}, \mathcal{C}_{k-1}$ 
appear in turn. Note
  that for a point $x$ in the region $\REG(R)$, its fiber $f_x$ is an
  infinite ray because $R$ is convex. For $i=1,2,\ldots,k-1$, consider
  the half-line $x_i\phi_i$ from a point at infinity
  $\phi_i\in\mathcal{C}_i$ to a point $x_i\in R$ closest to
  $\phi_i$. If $x \in x_i\phi_i$ lies in $\REG(R)$, then
  $f_x\subset\REG(R)$, which is impossible because
  $\phi_i\in\mathcal{C}_i$. Define $W=R\cup\bigcup_i{x_i\phi_i}$.
  %% ($W$ looks like a hand-drawn sun)
  We then have $W\cap\REG(R)=\emptyset$, and $\R^2\setminus W$
  consists of $k-1$ unbounded connected subsets of the plane.  The
  connected subset bounded by $x_i\phi_i$, $x_{i+1}\phi_{i+1}$ and
  $R$ contains $\mathcal{R}_i$ completely. It follows that
  $\mathcal{R}_i \neq \mathcal{R}_j$ when $i\neq j$.

It remains to show that  the bound of $k-1$ is tight. 
By Lemma~\ref{lem:linear-at-infinity}, there are at most $2k-2$ pure
vertices at infinity, and thus at most $2k - 2$ unbounded Voronoi
components. Hence, one single Voronoi region has at most $k-1$ unbounded Voronoi
components (since two neighboring components cannot correspond to the same site). 
This concludes the proof because, by
Lemma~\ref{lem:bounded-connected}, if a Voronoi region has a bounded component, then the entire region is connected.
\end{proof}

\section{Algorithm} %Computing the Voronoi diagram}
\label{sec:computing}

The proof of Theorem~\ref{thm:size} suggests an algorithm for computing the
Voronoi diagram by sweeping the arcs of the graph $G$. This is roughly
equivalent to computing the surface $\UDIST$ by sweeping a horizontal plane
downwards, and maintaining the part of $\UDIST$ above this plane. This is
essentially the approach used by Aurenhammer~\etal~\cite{adk-flsvd-06} for the
computation of farthest-segment Voronoi diagrams. However, this does not seem
to work for our diagram because of the mixed vertices of 
\referee{Referee \#1: 'of type (f)' you refer to Figure 3 here?}
\respond{Yes. Clarified.}
type~(f) (see~Fig.~\ref{fig:vertices}), where
$\UDIST$ has a local maximum.
We instead offer a divide-and-conquer algorithm.   
\begin{theorem}
  \label{thm:main}
  The farthest-polygon Voronoi diagram $\FS$ of a family
  $\SITES$ of disjoint polygonal sites of total complexity $n$ can be
  computed in time $O(n \log^3 n)$.
%  \complain{Sam says: I removed ``expected'' [time] here.}
\end{theorem}
\begin{proof}
  Let $\SITES =\{P_{1},\dots, P_{k}\}$, and let $n_{i}$ be the
  complexity of~$P_{i}$.  If $k = 1$, then $\FS$ is simply
  the medial axis $\MA(P_{1})$, which can be computed in time $O(n
  \log n)$~\cite{f-savd-87}.  Otherwise, we split $\SITES$ into two
  disjoint families $\SITES_{1}, \SITES_{2}$ as follows:
  \begin{di}
    \item If there is a site $P_{i}$ with complexity $n_{i} \geq n/2$,
    then $\SITES_{1} = \{P_{i}\}$ and $\SITES_{2} =
    \SITES\setminus\{P_{i}\}$.
    \item Otherwise there must be an index $j$ such that $n/4 \leq
    \sum_{i=1}^{j} n_{i} \leq 3n/4$.  We let $\SITES_{1} =
    \{P_{1},\dots, P_{j}\}$ and $\SITES_{2} = \{P_{j+1}, \dots,
    P_{k}\}$.
  \end{di}
  We recursively compute $\FI$ and $\fpvd(\SITES_{2})$.
  We show in the rest of this section (see Lemma~\ref{lem:merging}) that we can
  then merge these two diagrams to obtain $\FS$ in time $O(n\log^{2} n)$,
  proving the theorem. 
\end{proof}

%It remains to discuss the merging step.  We are given  a 
%family $\SITES$ of disjoint polygonal sites of total complexity $n$, 
%and we are given $\FI$ and $\FII$, where $\SITES = \SITES_{1} \cup
%\SITES_{2}$ is a disjoint partition of~$\SITES$.

We now discuss  the merging step.
We are given the farthest-polygon Voronoi diagrams $\fpvd(\SITES_1)$ and $\fpvd(\SITES_2)$ of  two families
$\SITES_1$ and $\SITES_2$ of pairwise disjoint polygonal sites, and let $\SITES = \SITES_{1} \cup
\SITES_{2}$.

%We consider given a family $\SITES$ of disjoint polygonal sites of total complexity $n$, 
%and $\FI$ and $\FII$, where $\SITES = \SITES_{1} \cup
%\SITES_{2}$ is a disjoint partition of~$\SITES$. 

Consider the diagram $\FS$ to be computed.  We color the Voronoi
regions of $\FS$ defined by sites in $\SITES_{1}$ red, and the Voronoi
regions defined by sites in $\SITES_{2}$ blue.  A pure arc of $\FS$ is red
if it separates two red regions, and blue if it separates two blue
regions.  The remaining pure arcs, which separate a red and a blue
region, are called~\emph{purple}.  A vertex of $\FS$ is purple if
it is incident to a purple arc.  We observe (see Fig.~\ref{fig:vertices}) that, by our general
position assumption, every  purple vertex not at infinity is incident to exactly
two purple arcs, and so the purple arcs form a collection of bounded and unbounded chains (see Fig.~\ref{fig:algo}). 

As we will see in Lemma~\ref{lem:merging}, merging $\FI$ and $\FII$ can be done in linear time once all purple
arcs are known, because the diagram $\FS$ consists of those portions
of $\FI$ lying in the red regions of $\FS$, and those portions of
$\FII$ lying in the blue regions of~$\FS$, see Fig.~\ref{fig:algo}.
\begin{figure}[t]
  \centering
  \includegraphics{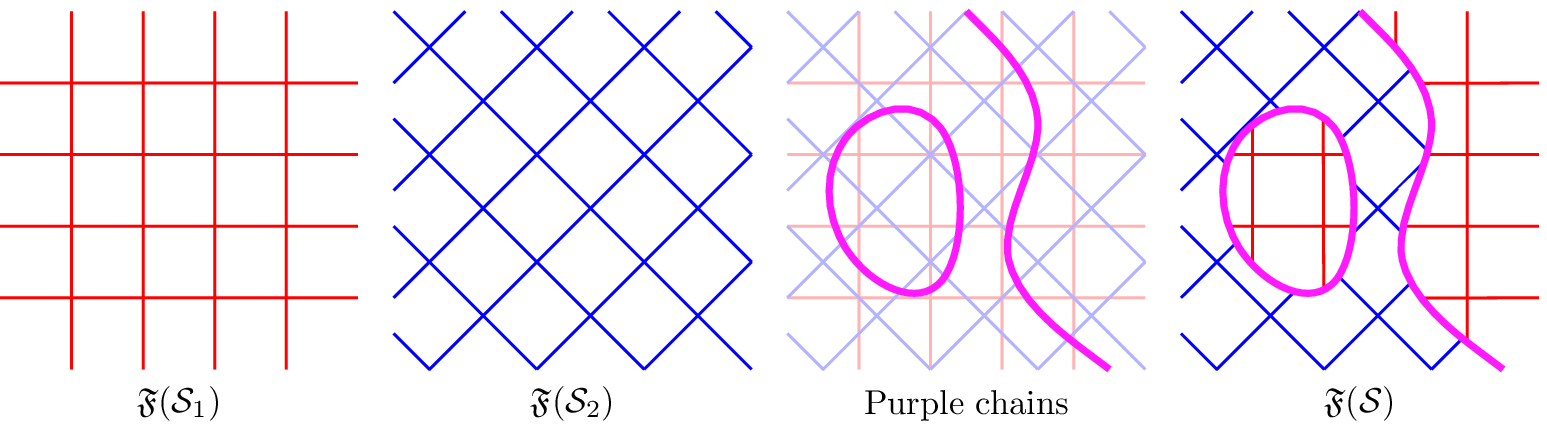}
  \caption{Merging $\FI$ and $\FII$ to obtain $\FS$.}
  \label{fig:algo}
\end{figure}

We show below how the purple chains of $\FS$ can be computed in time
$O(n\log^2 n)$. We first show how to compute at least one point on every chain
and then how to ``trace'' a chain from a starting point.

%We first show how to compute at least one point on every chain. 
% 
%We show below that if we have a starting vertex on every purple chain, then
%we can trace the purple chains in total time~$O(n)$.	For the open chains, we
%can use the purple vertices at infinity as starting vertices, as these are
%easy to compute.	For the closed purple chains, we make use of the following
%lemma:

\begin{lemma}
  \label{lem:pure-inifinity-algo}
The vertices at infinity of the purple chains can be computed in time $O(n \log n)$.
%The pure vertices of $\FS$ at infinity can be computed in time $O(n \log n)$.
\end{lemma}
\begin{proof}
We show how to compute all the pure vertices at infinity in time $O(n \log
n)$. 
%Then, we can easily deduce the purple ones  from these $O(k)$ pure vertices (by
%Lemma~\ref{lem:linear-at-infinity}), in $O(k)$ time. 
There are $O(k)$ such pure vertices (by Lemma~\ref{lem:linear-at-infinity}), and we can easily
deduce the purple ones  from those in $O(k)$ time. 
\SL{I modified the above sentence which was "Then, we can easily deduce the purple ones  from these $O(k)$ pure vertices (by
Lemma~\ref{lem:linear-at-infinity}), in $O(k)$ time." because I found it somehow quite confusing (it
sounded like the fact that we can "deduce" came from Lemma~\ref{lem:linear-at-infinity} which made no
sense). This was not specifically asked by the reviewer.}

For site $P_{i}\in \SITES$ and angle $\phi \in [0, 2\pi)$, let
  $\ell_{i}(\phi)$ be the oriented line with direction~$\phi$ tangent
  to~$P_{i}$ and keeping $P_{i}$ entirely on its left, see
  Fig.~\ref{fig:lower-envelope}. Let $g_{i}(\phi)$ be the signed
  distance from the origin to $\ell_{i}(\phi)$ (positive if the origin
  lies left of $\ell_{i}(\phi)$, negative otherwise).  If $P_{i}$ is a
  polygonal site of complexity $m$, we first compute the convex hull
  $\conv(P_{i})$ in time $O(m \log m)$, and we deduce  a
  description of the function $g_{i}$ in time~$O(m)$.

We then compute the lower envelope $g$ of the functions~$g_{i}$.  The
pure vertices at infinity correspond exactly to the breakpoints of
this lower envelope, since they correspond to half-planes (or disks with centers at infinity)
touching two sites and whose interiors intersect all other sites. Such a half-plane is
illustrated in gray in Fig.~\ref{fig:lower-envelope}. 
\referee{Referee \#1: Proof of Lemma 12, 2nd paragraph. I can see that the breakpoints of the lower
envelope of the $g_i$ is important in some way, but I don't see the 'half planes...whose interiors
intersect all other sites...'? Figure 6a does not help me in any way here. Also, I don't understand
what Lemma 4 has to do with the fact that $g_i$'s can pairwise intersect only twice, After all, they
are quadratic functions/parabolas (?) that are open on top/or the bottom (depending on the sign) and
that follows straight from there?}
\respond{Clarified the following sentence and improved the illustration in Figure 6a. Note that the
last comment is wrong because the functions $g_i$ are not quadratic functions, but piecewise quadratic.}
Moreover, two functions $g_{i}$, $g_{j}$ can
intersect at most twice since each intersection corresponds to a pure vertex at infinity of 
$\fpvd(\{P_i, P_j\})$ which admits at most two such vertex as argued in the
proof of Lemma~\ref{lem:linear-at-infinity}. 
Hence, the lower envelope can be computed in time $O(n \log n)$~\cite{sa-dsstg-95}.
%
%Since two functions $g_{i}$, $g_{j}$ can
%intersect at most twice (see the proof of
%Lemma~\ref{lem:linear-at-infinity}), the lower envelope can be
%computed in time $O(n \log n)$~\cite{sa-dsstg-95}.
  \end{proof}

%\begin{figure}[t]
%  \centering
%  \includegraphics{fig/lower-envelope}
%  \caption{Constructing the vertices at infinity.}
%  \label{fig:lower-envelope}
%\end{figure}

\begin{figure}[t]
\centering
\subfigure[]{\includegraphics{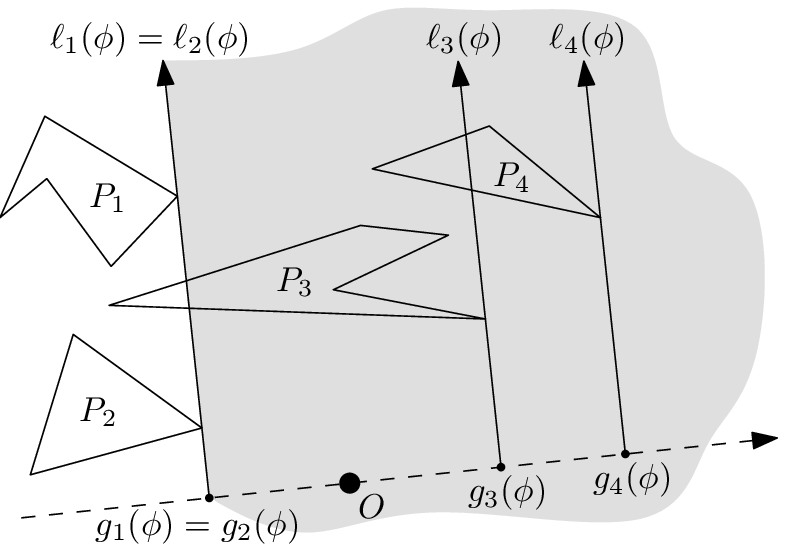}\label{fig:lower-envelope}}%
\hfill
\subfigure[]{\includegraphics{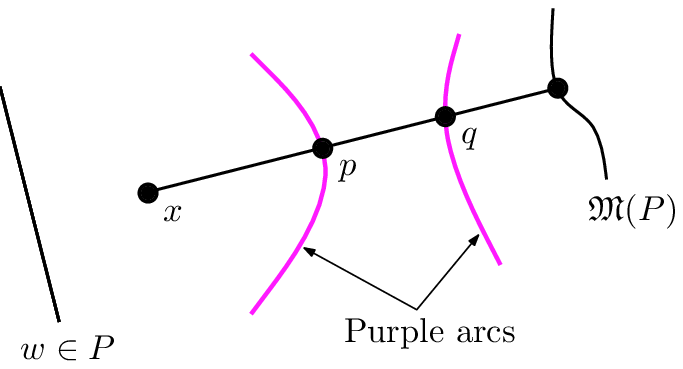}\label{fig:lem-tracing}}%
\caption{(a) Constructing the pure vertices at infinity. (b) For the proof of
Lemma~\ref{lem:tracing}: a fiber $f_x$ cannot have more than one intersection point with purple arcs.}
 \end{figure}
\referee{Referee \#2: Caption of Fig 6b: $->$ "a fiber $f_x$ cannot have more than one intersection
points with purple arcs."}
\respond{Done.}

\begin{lemma}
  \label{lem:mixed}
  Any bounded purple chain contains a mixed  vertex  of $\FS$.
\end{lemma}
\begin{proof}
%\SL{I rewrote this proof which was too fast for me (and for Joachim, it seems) and I added the old
%Lemma 1 in the proof.}
A bounded purple chain is a compact set in the plane, and so the restriction of
$\UDIST(\cdot)$ to such a chain admits a maximum. Such a maximum appears at
a vertex, denoted $v$. Indeed, since an arc $\alpha$ of $\FS$ is defined by two
features (corners or edges), the graph of $\UDIST(\cdot)$ restricted to
$\alpha$ is the intersection of the two Voronoi surfaces---which are cones or
wedges---induced by the two features; thus $\UDIST(\cdot)$ cannot have a
local maximum in the interior of~$\alpha$ (it may have a local minimum).

Consider now the arcs of $\FS$ oriented, in a neighborhood of their endpoints,
in the direction of increasing $\UDIST(\cdot)$. Then, the two purple arcs
incident to $v$ point toward $v$. Now observe that purple arcs are pure and a
vertex incident to at least two pure arcs pointing toward it is of type (e) or
(f), which is mixed (see Fig.~\ref{fig:vertices}). Hence, $v$ is a mixed vertex
of $\FS$, which concludes the proof.
%A bounded purple chain is a compact set in the plane, and so the restriction of $\UDIST(x)$ to such a
%chain admits a maximum.  By Lemma~\ref{lem:local-maximum}, a local maximum of the restriction of
%$\UDIST(x)$ to the arcs of $\FS$ can only appear at vertices. Thus, the restriction of $\UDIST(x)$
%to a bounded purple chain admits a local maximum at one of its vertices, say $v$.  Consider now the
%arcs of $\FS$ 
%oriented in the direction of increasing $\UDIST(x)$ in a neighborhood of their endpoints. 
%Then, the two purple arcs incident to $v$ point toward $v$. Now observe that purple arcs are pure and a vertex incident to at least two pure arcs pointing toward it is of type (e) or (f), which is mixed (see Fig.~\ref{fig:vertices}). 
%%oriented, in a neighborhood of their endpoints, in the direction of increasing
%%$\UDIST(x)$. Then, the two purple arcs incident to $v$ point toward $v$.  Finally, recall that
%%purple arcs are pure, and observe (see Fig.~\ref{fig:vertices}) that, if a vertex of $\FS$ admits at
%%least two pure arcs pointing toward it, then the vertex is of type (e) or (f), which is mixed.
%Hence, $v$ is a mixed vertex of $\FS$, which concludes the proof.
\end{proof}

\begin{lemma}
  \label{lem:mixed-algo}
Given the farthest-polygon Voronoi diagrams of  two families
$\SITES_1$ and $\SITES_2$ of pairwise disjoint polygonal sites, the 
  mixed vertices  of $\fpvd(\SITES_1\cup\SITES_2)$ can be computed in time $O(n \log^2 n)$. 
\end{lemma}
Computing the mixed vertices in time $O(n \log^2 n)$ is the most subtle part
of the algorithm and we postpone the proof of this lemma to
Section~\ref{sec:mixed}, after showing how to compute the purple chains from
some starting point (Lemma~\ref{lem:purple-chains-algo}).
We start with a preliminary lemma which
is an important consequence of Lemma~\ref{lem:fiber}.

\begin{lemma}
  \label{lem:tracing}
  Let $f_{x}$ be the fiber of point $x$ in a cell $C$ of $\FI$ or $\FII$.
  Then, the relative interior of 
  $f_{x}$  intersects the  purple arcs of~$\FS$ in at most one point. 
\end{lemma}
%\begin{figure}[thb]
%  \centering
%  \includegraphics{fig/lem-tracing}
%  \caption{For the proof of Lemma~\ref{lem:tracing}: a fiber $f_x$ cannot intersect two purple arcs.}
%  \label{fig:lem-tracing}
%\end{figure}
\begin{proof}%\SL{I expended this proof}
  Assume, for a  contradiction, that the relative interior of $f_x$ intersects
  two purple arcs in two distinct %\SL{Added "two distinct"}
points $p$ and $q$, where $q$  lies on~$f_{p}$ (see Fig.~\ref{fig:lem-tracing}). 
%%  Assume, for a contradiction, that two purple arcs intersect the relative interior of $f_{x}$ in
%%  points $p$ and $q$, where $q$ lies on~$f_{p}$ (see Fig.~\ref{}).  
%We can assume, without loss of
%  generality that $p$ is not a point of tangency between $f_x$ and the purple chain, because,
%  otherwise, there exists another fiber close to $f_x$ that intersects transversally the purple
%chain in two points. (Consider for example a fiber $f_y$ with $y$ close to $x$.)
\referee{Referee \#1: Proof of Lemma 15: ' w.l.o.g. p is not a point of tangency between $f_x$ and the
purple chain...'. Maybe it's better to write that the purple chain and $f_x$ cross.}
\referee{Referee \#2: p10, l-7: If x* happens to be a corner of P, there might not be another fiber
close to $f_x$ that intersects the purple chain twice. Does the general position assumption help?}
\respond{The previous version of the following sentence troubled both referees and we rewrote it
following essentially the suggestion of the first referee.}
\SL{We had  "We can assume, without loss of
  generality that $p$ is not a point of tangency between $f_x$ and the purple chain, because,
  otherwise, there exists another fiber close to $f_x$ that intersects transversally the purple
  chain in two points." This puzzled two 2 referees. The second referee seems partiall wrong or at least I don;t see the connection with x* being a
  corner. However, I believe that if the purple chain is tangent and cross $f_x$, our sentence was
  wrong or at least we missed a S in  "... that intersects transversally the purple
  chainS in two points.". Anyway, it seems that the following sentence should/might satisfy both reviewers
  and that the proof remains correct after that. AM I MISSING A REASON WHY WE GOT RID OF TANGENCY
  RATHER THAN CROSSING?}
We can assume, without loss of
  generality that the purple chain and $f_x$ cross at $p$, because,
  otherwise, there exists another fiber close to $f_x$ (for instance, $f_y$ for some $y$ close to
  $x$) that intersects transversally the purple chains in two points. 
Now, let $P$ be the site containing the feature $w$
  associated with $C$.  In $C$, the purple arcs bound the cell of $\FS$ belonging to feature $w$.
  Hence, there is a point $p'$ on $f_{x}$ sufficiently close to $p$ such that $p' \in \REG(P)$ in $\FS$.
%  \SL{This assumes that $f_x$ is not tangent to the purple chain.} 
Moreover, the fiber $f_{p'}$ is,
  by definition, a subset of the fiber $f_x$. Thus, the fiber $f_{p'}$ contains $q$ and thus
  intersects a purple arc, contradicting the fact that $f_{p'}$ lies in $\REG(P)$ in $\FS$ (by
  Lemma~\ref{lem:fiber}).
%
%In the neighborhood of $p$, the  purple chain through $p$ bounds the component $\REG(P)$ in $\FS$. 
%But then the fiber $f_{p'}$, which is by construction a subset of the fiber $f_x$, 
%  lies in $\REG(P)$ in $\FS$, and cannot intersect a purple arc at~$q$.
\end{proof}

\begin{lemma}
  \label{lem:purple-chains-algo}
Given the farthest-polygon Voronoi diagrams $\fpvd(\SITES_1)$ and $\fpvd(\SITES_2)$ of  two families $\SITES_1$ and $\SITES_2$ of pairwise disjoint polygonal sites,  the  purple chains of $\fpvd(\SITES_1\cup\SITES_2)$ can be computed in time $O(n \log^2 n)$. 
\end{lemma}
\begin{proof}
\referee{Referee \#2: p11, l3 and p13, l16: I would prefer "Lemmata".}
\respond{This is really a question of style and we prefer "Lemmas"}
  By Lemmas~\ref{lem:pure-inifinity-algo}, \ref{lem:mixed} and~\ref{lem:mixed-algo}, we can compute the vertices at
  infinity of the purple chains, and a superset of size $O(n)$ of at least one mixed vertex per bounded
  component. As we have seen in the proof of Lemma~\ref{lem:mixed}, these latter vertices are of
  type (e) or (f) (see Fig.~\ref{fig:vertices}); they thus involve only two sites  and we 
  discard all those that involve two sites of $\SITES_1$ or two sites of $\SITES_2$. We thus obtain
  a set of $O(n)$ purple vertices with at least one such vertex on each purple chain. 
%We have  a description of the purple arc(s) incident to these vertices but we do not know their other
%  endpoints. 
We consider each such vertex, denoted $v$, in turn, and ``trace'' (construct) the purple chain that $v$ lies on.

If $v$ is at infinity, there is only one purple arc incident to it; otherwise, there are two and we consider one of them.
The bisector supporting the incident purple arc is that of one red and one
blue feature among the features defining $v$. Splitting the bisector at $v$ defines two semi-infinite
curves incident to $v$ and we can determine, using the respective position of the three features defining $v$,
which of these two curves supports the considered purple arc; %incident to $v$;
we call this semi-infinite curve a purple half-bisector.

We then trace the purple chain by following its purple arcs from cell to cell: %, as follows. 
observe that the other endpoint of the considered purple arc incident to $v$ is either at infinity
or is the first
intersection point (starting from $v$) between the purple half-bisector
and the cell boundaries of  $\FI$ and~$\FII$.
We compute this point as follows.
%We compute, as follows, the first point of intersection (if any) of the purple half-bisector with
%the cell boundaries of $\FI$, and similarly for $\FII$.

We first locate $v$ in $\FI$ and $\FII$. This can be done in $O(\log n)$ time, assuming that
we have precomputed a point-location data structure  for  $\FI$ and $\FII$
in $O(n\log n)$ time (see, for instance,~\cite{edelsbrunner86}).%
\,\footnote{In fact, we will see in Section~\ref{sec:mixed} that the
combinatorial description of $v$ and all the information regarding its location
in $\FI$ and $\FII$ are already available as a by-product of the computation of
the mixed vertices. This makes the location procedure described here not
strictly necessary.}
%\SL{Sort out the point location data structure and give a ref here.} 
Note that $v$ is
%, initially, 
a vertex of $\FS$ and that it involves two features of one site, say in
$\SITES_1$, thus $v$ lies on a medial axis arc of $\FI$ induced by these two features. 
We then also determine the cell of $\FI$, denoted $C_1$, that contains the purple half-bisector in a
neighborhood of $v$ (this is a constant size problem).
Since
$v$ is a purple vertex, its third feature belongs to a site of $\SITES_2$, and $v$
lies in the cell of $\FII$, denoted $C_2$, belonging to that feature.
Let $w_i$ denote the feature associated with $C_i$, $i=1,2$.%

\referee{Referee \#2: p11, l24 "We then determine whether the half-bisector goes left or right..."
This should be stated more precisely....}
\referee{Referee \#1: Proof of Lemma 16: page 11 line 41: 'We have thus computed one point of
intersection OF THE PURPLE HALF-BISECTOR with the boundary...and another one WITH the "}
\referee{Referee \#1: page 11 line 45/46: Why is a sweep 'in parallel' necessary for the time complexity? I don't see this. Also, (line 50 onwards) it seems that you don't really go parallel, but that you are jumping between cells $C_1$ and $C_2$ with the sweep, depending on which one seems more 'promising' at the moment? Please clarify.}
\referee{Referee \#2: p12, l1-13: ...because the time analysis depends crucially on the monotonicity property. One cell can be visited several times by different purple chains. Also, it seems that a cell can be split into several subcells during the merge. Given these complications, I think you need to argue more clearly why multiple work can be avoided thanks to monotonicity. Perhaps some illustrations would be helpful.}

\respond{The rest of this proof has been substantially  rewritten to clarify these issues.}

%Because the cells of the diagrams $\FI$, $\FII$ and $\FS$ are monotonous, as
%defined in Lemma~\ref{lem:monotonicity2}, we can sweep any cell with a line
%that intersects the cell in a single line segment, which we call a
%\emph{sweep-line-segment}. When the cell's feature is an edge, the sweep-line
%translates; when it is a corner, the sweep-line contains and rotates around it.
%Informally, we ``attach'' both sweep-lines in $C_1$ and $C_2$ to a single point
%$x$ moving on the purple half-bisector away from $v$. During the motion of $x$,
%we maintain the endpoints of both sweep-line-segments, one in each cell $C_1$
%and $C_2$. This combinatorial information, together with the geometrical
%description of $\FI$ and $\FII$ is enough to trace the full purple chain.

Then, for $i=1,2$, we 
 sweep the cell $C_i$ with a line orthogonal to the feature $w_i$
if $w_i$ is an edge and with a line through the feature $w_i$ if $w_i$  is a corner. 
If $w_i$ is a corner, the sweep  is done clockwise or counterclockwise so that the sweep line
intersects the half-bisector inside $C_i$ 
(deciding between clockwise or
counterclockwise is a constant size problem);
%(determining whether it should be clockwise or
%counterclockwise is a constant size problem involving, in a neighborhood of $v$, the boundary of $C_i$
%and the purple half-bisector); 
the situation is similar when $w_i$ is an edge.
%similarly, if $w_i$ is an edge, the sweep  is performed in the direction so
%that the  sweep line intersects the half-bisector in $C_i$.

The two cells $C_1$ and $C_2$ are swept simultaneously. 
However, since the two sweeps are not a priori performed using the same sweep line, this
requires some care.  For clarity, we first present each sweep independently.

By Lemma~\ref{lem:monotonicity2}, the sweep line always intersects one arc of the upper chain of
$C_i$ (or two arcs at their common endpoints) and similarly for the lower chain.  We first determine
the arcs of the upper and lower chains 
%the arc on the lower chain and the one on the upper chain 
that are intersected by the sweep line
through $v$ (or about to be intersected if the sweep line goes through a vertex of the chain).  We
also determine the intersection, if any, of these two arcs with the purple half-bisector.  When the
sweep line reaches an endpoint of one of the two arcs that are being swept, we determine the
intersection (if any) between the new arc and the purple half-bisector.  When the sweep line reaches
the first of the computed intersection points between the purple half-bisector and the boundary of
the cell, we report this intersection point and terminate the sweep of $C_i$.

If the two sweeps were performed independently, we could report the first point where  the purple half-bisector
exits one of the cells $C_1$ or $C_2$, and continues the tracing in a neighboring cell of
either $\FI$ or $\FII$, along a new purple half-bisector. This however would not yield the claimed
complexity because, roughly speaking,  if  a sequence of purple arcs enter and
exit $\Theta(n)$ cells  $C_1,C_1',C_1'',\ldots$ while remaining inside a cell $C_2$ of complexity
$\Theta(n)$, the cell $C_2$ would be swept many times, possibly leading to a complexity of $\Theta(n\log n)$
per arc, and a total of  $\Theta(n^2\log n)$. 
We thus perform the  two sweeps  simultaneously, as follows.

Note first that, by Lemma~\ref{lem:tracing}, during the sweep of $C_i$, the point of intersection,
in $C_i$, between the sweep line and the purple half-bisector moves \emph{monotonically} along the
purple half-bisector. % (and away from $v$). 
We can thus parameterize the sweeps of $C_1$ and $C_2$ 
by a point $x$ moving monotonically on the purple half-bisector away from $v$. The point $x$ define 
two sweep lines (the lines through $x$ and through $w_i$ or orthogonal to $w_i$ depending on the
nature of $w_i$), and  the events are those of the sweeps of $C_1$ and $C_2$. 
This sweep ends when the purple half-bisector leaves $C_1$ or $C_2$.
Then, the tracing of the purple chain continues in a neighboring cell of
either $\FI$ or $\FII$, along a new purple half-bisector.
We stop tracing the purple chain when we reach a vertex at infinity or the
vertex $v$ we started from. 

We then consider a new starting point $v$; note that 
we can easily check in $O(\log n)$ time whether it has already been computed
(while tracing some purple chains)
by maintaining the list of the already computed vertices on the purple chains,
ordered lexicographically by their features.

We now analyze the complexity of the algorithm. Consider first the initialization of every sweep,
that is the determination of the arcs of $C_1$ and $C_2$ that are intersected by the two sweep
lines through $v$. There are $O(n)$ such initialization steps to perform, since $\FS$ have size
$O(n)$ by Theorem~\ref{thm:size}, and each step can be done using binary search in $O(\log n)$ time, after
preprocessing all the cells of $\FI$ and $\FII$; the preprocessing (which simply is storing the ordered
vertices of the upper and lower chains of each cell in arrays) can be done in time linear in the total size of
the cells, which is $O(n)$, by Theorem~\ref{thm:size}.

We finally analyze the complexity of the rest of the algorithm by applying a simple charging
scheme. 
Note first that intersecting the purple half-bisector with an arc
of a cell takes constant time. We charge the cost of computing these
intersections
% between the purple half-bisector and the arcs of $C_i$ to 
%either these arcs of to the purple vertices. 
to either the purple vertices or to the arcs of $C_i$, as follows. The intersections with
each of the arcs that are swept at the beginning and the end of the sweep are
charged to the corresponding endpoint of the purple arc. Every purple vertex
is thus charged at most eight times because each of the two incident purple edges is intersected
with one arc of each of the
lower and upper chains of each of the two cells $C_1$ and $C_2$.
%twice for each of the two sweeps and
%each of the two incident purple edges. 
The intersections with each of the
other arcs of $C_i$ are charged to the arc in question. Every such arc is charged at most once per cell, and
thus at most twice in total; indeed, such an arc of $C_i$ is swept entirely during the
sweep of $C_i$ and thus, by Lemma~\ref{lem:tracing}, it will not be swept again during
another sweep of cell $C_i$ (when treating another purple half-bisector).
Since $\FI$, $\FII$, and $\FS$ have size $O(n)$ (Theorem~\ref{thm:size}), all
the purple arcs can be computed in $O(n\log n)$ time, in total, once the set of
starting points are known, and thus in $O(n\log^2n)$ time by
Lemmas~\ref{lem:pure-inifinity-algo} and~\ref{lem:mixed-algo}.
%
%We now turn to the time complexity of tracing the purple chains in $\FS$.
%First, consider a sweep-line-segment as a \emph{maximal fiber} in the sense
%that it is the limit of an increasing (w.r.t. inclusion) sequence of fibers.
%Then, all sweep-line-segments that appear when tracing the purple chains are maximal
%fibers in the cell that they sweep. Hence, from Lemma~\ref{lem:tracing}, we get
%that any maximal fiber (or sweep-line-segment) is considered at most once when
%tracing the purple chains, and the number of sweep events is then proportional
%to the sum of the sizes of $\FI$, $\FII$ and $\FS$.  We finally observe that
%handling a sweep event takes constant or $O(\log n)$ time.
%%
%Since $\FI$, $\FII$, and $\FS$ have size $O(n)$ (Theorem~\ref{thm:size}), all
%the purple arcs can be computed in $O(n\log n)$ time, in total, once the set of
%starting points are known, and thus in $O(n\log^2n)$ time by
%Lemmas~\ref{lem:pure-inifinity-algo} and~\ref{lem:mixed-algo}.
\end{proof}

\begin{lemma}
  \label{lem:merging}
  Given the farthest-polygon Voronoi diagrams of two families
  $\SITES_1$ and $\SITES_2$ of pairwise disjoint polygonal sites, the
  farthest-polygon Voronoi diagrams of $\SITES_1\cup\SITES_2$ can be
  computed in time $O(n \log^2 n)$.
\end{lemma}
\begin{proof}
  During the computation of the purple chains (see the proof of
  Lemma~\ref{lem:purple-chains-algo}), we can split the arcs of $\FI$
  and $\FII$ at every new purple vertex that is computed. Then,
  merging $\FI$ and $\FII$ can trivially be done in linear time since,
  as we mentioned before, the diagram $\FS$ consists of those portions
  of $\FI$ lying in the red regions of $\FS$, and those portions of
  $\FII$ lying in the blue regions of~$\FS$, see Fig.~\ref{fig:algo}.
\end{proof}

\subsection{Computing the mixed vertices}
\label{sec:mixed}

We prove here Lemma~\ref{lem:mixed-algo} stating that, given the
farthest-polygon Voronoi diagrams of two families $\SITES_1$ and $\SITES_2$ of
pairwise disjoint polygonal sites, the mixed vertices of
$\fpvd(\SITES_1\cup\SITES_2)$ can be computed in time $O(n \log^2 n)$.

We start by computing the randomized point-location data structure of
Mulmuley~\cite{m-fppa-90} (see also \cite[Chapter~6]{bkos-cgaa-08})
for the two given Voronoi diagrams~$\FI$ and~$\FII$. (We chose this
randomized algorithm for its simplicity; randomization can however be
avoided as we explain at the end of this section.)  This data
structure only needs two primitive operations: (i) for a given point
$p$ in the plane, determine whether the query point lies left or right
of~$p$, and (ii) for an $x$-monotone line segment or parabolic
arc~$\gamma$, determine whether the query point lies above or
below~$\gamma$.  Both cases can be summarized as follows: Given a
\emph{comparator} $\gamma$, determine on which side of $\gamma$ the
query point lies.  The comparator can be either a line or a parabolic
arc.

\referee{Referee \#1: page 12, line 35, 38, 61: parabola arc - PARABOLIC arc?}
\respond{Done.}

We compute the mixed vertices lying on each tree $\TREE$ of each medial axis $\MA(P)$ separately.
%For each tree $\TREE$ of $\MA(P)$, the
The intersection $\TREE \cap \REG(P)$ is a
connected subtree by Lemma~\ref{lem:tree}. 
\referee{Referee \#2: Perhaps "locate" would be better than "determine".}
\respond{Done}
We can locate the internal
vertices of this subtree easily, by performing a point location operation for
each vertex $v$ of $\TREE$ in $\FI$ and $\FII$, deducing which site is
farthest from~$v$ and checking %$v$ lies in $\REG(P)$ if and only
if the farthest site
from $v$ is $P$. Let $I$ be the set of vertices of $\TREE$ that lie in
$\REG(P)$. We now need to consider two cases.

\subsubsection{When $I$ is not empty}

If $I$ is non-empty, then every arc $\alpha$ of $\TREE$ incident to
one vertex in~$I$ and one vertex not in~$I$ must contain exactly one
mixed vertex $s\s$ by Lemma~\ref{lem:tree}. To locate this vertex
$s\s$, we use \emph{parametric search} along the
arc~$\alpha$~\cite{m-apcad-83}, 
\referee{Referee \#2: Since you are referring to Megiddo's original paper, you should include some
remarks on the parallel algorithm you want to employ (or cite a black box approach, e.g., van
Oostrum/Veltkamp).}
\respond{Clarified as follows.}
albeit in a restricted way, as we do not need to use
parallel computation: The idea is to execute two point
location queries in $\FI$ and~$\FII$ using $s\s$ as the query point.
Each query executes a sequence of primitive operations, where we
compare the (unknown) location of $s\s$ with a comparator $\gamma$ (a
line or a parabolic arc).  This primitive operation can be implemented
by intersecting $\alpha$ with $\gamma$, resulting in a set of at most
four points.  In $O(\log n)$ time, we can test for each of these
points whether it lies in $\REG(P)$.  This tells us between which of
these points the unknown mixed vertex $s\s$ lies, and we can answer
the primitive operation.

It follows that we can execute the two point location queries on $s\s$
in time~$O(\log^{2} n)$, and we obtain one cell and one arc of $\FI$
and $\FII$ containing $s\s$ (for instance if $P\in\SITES_1$ then $s\s$
lies on an arc of $\FI$ and in a cell of $\FII$). The mixed vertex
$s\s$ lies at equal distance of the three features to which the arc
and cell belong (two of which are features of~$P$).

\subsubsection{When $I$ is empty}

It remains to consider the case where $I$ is empty, that is, no vertex
of $\TREE$ lies in $\REG(P)$. Nevertheless, the region $\REG(P)$ may
intersect a single arc $\alpha$ of $\TREE$. 
\referee{Referee \#1: '..and there are then two mixed vertices on $\alpha$ that we need to find.' Why are there two? Please explain.}
\respond{Clarified as follows}
In this case, Lemma~\ref{lem:tree}
tells us that there are two mixed vertices on~$\alpha$ that we need to find. We
first need to identify the arcs of $\TREE$ where this could happen.

Let $p$ and $q$ be two points on the same arc~$\alpha$ of $\MA(P)$.
We define the \emph{cylinder} $\cyl(p,q)$ of the pair $(p,q)$ as
\[
\cyl(p,q) = \bigcup_{x\in pq} D_{P}(x) \setminus D_{P}(p),
\]
where the union is taken over all points $x$ on the arc $\alpha$
between $p$ and~$q$.
Since $p$ and~$q$ belong to the same arc $\alpha$ of $\MA(P)$, which is induced by either two edges,
two vertices, or one edge and one vertex of $P$, cylinders may have only  three possible shapes,
illustrated in Fig.~\ref{fig:single-arc}(b, c).
%Fig.~\ref{fig:single-arc}(b, c) illustrate the three possible geometric shapes of a cylinder.
\begin{figure}[t]
  \centerline{\includegraphics{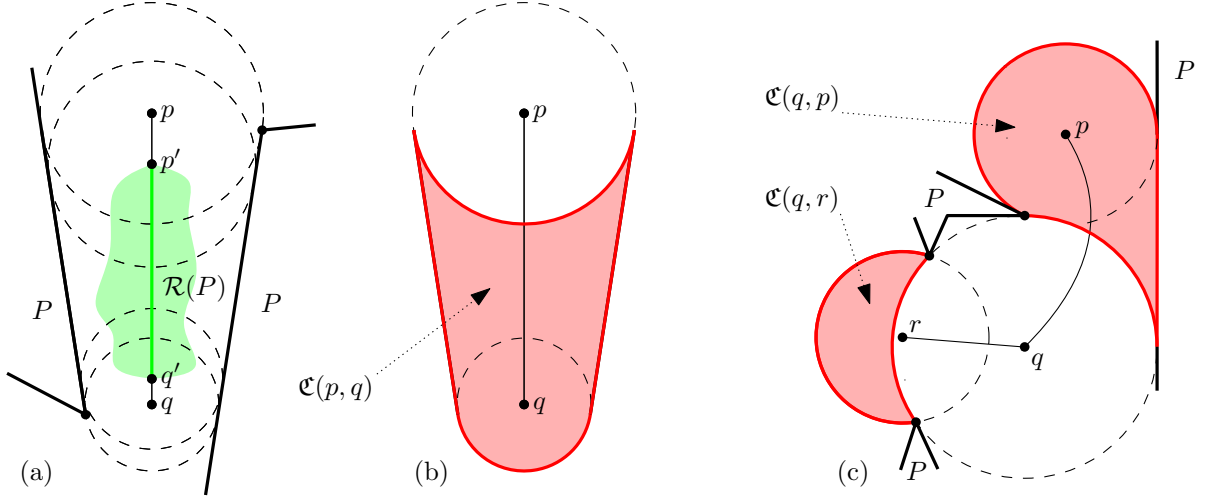}}
  \caption{(a) $pq$ is a sub-arc of $\MA(P)$. $p'q'$ is the
    intersection of $pq$ with $\REG(P)$. We also have $p'q'=\MA(P)
    \cap\REG(P)$. (b) The cylinder $\cyl(p,q)$ of the pair $(p,q)$
    in~(a). (c) Cylinders $\cyl(q,p)$ and $\cyl(q,r)$ illustrate the two other
    possible geometric shapes of a cylinder, whose boundary alternate between a
    part of a feature of $P$ and a circular arc.}
    %of the pair $(q,p)$ for a parabolic arc.}
  \label{fig:single-arc}
\end{figure}
We define a condition $G(p,q)$ as follows: Let $Q$ be a site farthest
from $p$, and let $w$ be a feature of $Q$ closest to~$p$.  Then
$G(p,q)$ is true if $w \subset \cyl(p,q)$ or if $Q$
intersects~$D_{P}(q)$.  Note that $G(p,q)$ could possibly depend on
the choice of~$Q$ and~$w$ when the farthest site or the closest
feature is not unique.  We will show below that this choice does not
matter for the correctness of our algorithm.

We can now prove the following two lemmas:
\begin{lemma}
  \label{lem:nc}
  Let $p, q$ be points on the same arc $\alpha$ of $\MA(P)$, such that
  neither $p$ nor $q$ lie in $\REG(P)$.  If $\alpha$ intersects
  $\REG(P)$ between $p$ and $q$, then $G(p,q)$ and $G(q,p)$ both hold.
\end{lemma}
\begin{proof}
  Since the statement is symmetric, we only need to show~$G(p,q)$. Let
  $Q$ be a site farthest from $p$, and let $x$ be a point between $p$
  and $q$ on $\alpha$ that lies in $\REG(P)$. This implies that $Q$
  intersects $D_{\SITES}(x) = D_{P}(x)$.  Since $Q$ does not intersect
  $D_{P}(p)$, we deduce that $Q$ intersects $\cyl(p,q)$.  If $Q$ does
  not lie entirely in $\cyl(p,q)$ then, since the sites $P$ and $Q$ are
  disjoint, $Q$ must cross the common boundary of $\cyl(p,q)$ and $D_{P}(q)$
  (see Fig.~\ref{fig:single-arc}(b, c)). Thus $Q$ intersects $D_{P}(q)$
  and indeed $G(p,q)$ holds.
\referee{Referee \#2: Figure 7,a+b, illustrate only one situation possible. It seems that a more
general statement is needed in the proof (perhaps that the boundary of C(p,q) always consists of
pieces of P and of the circular arcs belonging to $D_P(p)$ and $D_P(q)$.}
\respond{The issue is clarified by illustrating all possible shapes of a
cylinder in the Figure~\ref{fig:single-arc}(b, c) (see also modified caption).
And by giving more details in the end of the proof above.}
\end{proof}
\begin{lemma}
  \label{lem:queen}
  Let $p, q$ be points on the same arc $\alpha$ of a tree $\TREE$ of
  $\MA(P)$ that admits no vertex in $\REG(P)$. If neither $p$ nor $q$
  lie in $\REG(P)$ and both $G(p,q)$ and $G(q,p)$ hold, then all
  points in $\TREE\cap \REG(P)$ lie on~$\alpha$.
\end{lemma}
\begin{proof}
  Assume, for a contradiction, that there exists a point $x$ on
  $\TREE$ in $\REG(P)$ not between~$p$ and~$q$ on~$\alpha$; assume,
  without loss of generality, that $p$ lies on the path joining $x$
  and $q$ in $\TREE$.  Since $G(p,q)$ holds, there is a farthest
  site~$Q \neq P$ from~$p$ such that $Q$ intersects~$\cyl(p,q)$.
  Therefore there exists a point $y$ between~$p$ and~$q$ on $\alpha$
  such that $y$ is closer to~$Q$ than to~$P$.  Since $x\in \REG(P)$ by
  assumption, $x$ is closer to $Q$ than to~$P$. We have shown that $x$
  and $y$ are closer to $Q$ than to~$P$ but the point~$p$, which is
  between $x$ and $y$ on $\TREE$, is closer to $P$ than to~$Q$. This
  contradicts Lemma~\ref{lem:path}, and concludes the proof.
\end{proof}
Let us call an arc $\alpha$ connecting vertices $p$ and~$q$ of $\TREE$
a \emph{candidate arc} if $G(p,q)$ and $G(q,p)$ both hold.
Lemma~\ref{lem:queen} implies immediately that if there are two
candidate arcs, then $\TREE \cap \REG(P)$ is empty, and there are no
mixed vertices on~$\TREE$.

Since we have point-location data structures for $\FI$ and $\FII$, we
can test the condition $G(p,q)$ in time $O(\log n)$ for a given arc
$\alpha$ in $\MA(P)$ and two points $p,q \in \alpha$.  This allows to
identify all candidate arcs in $O(m \log n)$ time, where $m$ is the
complexity of~$\TREE$.  If there are zero or more than one candidate
arcs, we can stop immediately, as there are no mixed vertices
on~$\TREE$.

It remains to consider the case where there is a single candidate
arc~$\alpha$ in $\TREE$. We again apply parametric search, using an
unknown point $s\s$ in $\alpha \cap \REG(P)$ as the query
point. During the point-location query, we maintain an interval $pq$
on $\alpha$ that must contain $s\s$ (if $s\s$ exists at all). To
implement a primitive query, we must determine the location of $s\s$
with respect to a comparator~$\gamma$. If the current interval $pq$ on
$\alpha$ does not intersect $\gamma$, we can proceed immediately,
otherwise we get a sequence of points $p = x_{0}, x_{1}, x_{2}, \dots, x_{s} =
q$ on~$\alpha$, where $2 \leq s \leq 5$ 
\referee{Referee \#1: ' ...otherwise we get a sequence of points... where $2\leq s \leq 5$.' Where do you get this bound on s from?}
\respond{Clarified as follows}
(since there are at most four
intersection points $x_1, x_2, \dots, x_{s-1}$ between two curves of degree at
most two). We first test if some
$x_{i} \in \REG(P)$ in $O(\log n)$ time. If so, we abort the process,
and use the method discussed above to find the mixed vertices on the
arc between $p$ and $x_{i}$ and between $x_{i}$ and $q$. If no $x_{i}$
lies in $\REG(P)$, we test the conditions $G(x_{i}, x_{i+1})$ and
$G(x_{i+1}, x_{i})$ for each consecutive pair. If the conditions hold
for no pair or for more than one pair, we can stop immediately, as
there cannot be a point of $\REG(P)$ on $\alpha$. If there is exactly
one pair, we have found the location of $s\s$ with respect to the
comparator~$\gamma$, and we continue the point location query.

If both point location queries on $s\s$ in $\FI$ and $\FII$ terminate
without encountering a point in~$\REG(P)$, the current arc $pq$ of
$\alpha$ lies entirely on an arc of $\FI$ and in a cell of $\FII$
(assuming that $P\in\SITES_1$). Moreover, we get this arc and cell
from the two point location queries. The two mixed vertices $s\s$ on
the arc $pq$ are then both defined by the same three features that
define this arc and cell,\footnote{Three features may define two mixed
  vertices, for instance in the simple special case where the sites
  consist of a V-shaped polygonal site and a point.} and they can be
computed in constant time.

\subsubsection{Avoiding randomization}

\referee{Referee \#1: Section 4.1.3: You're going out of your way not to tell any details
here... What are the predicates? Left/Right/Up/Below? What set of other parameters is the polynomial
expression depending on?}
\respond{We find this section quite clear but we nonetheless added the following footnote and
parenthesis, which we hope might help.}
 
Finally, we argue that, instead of using a randomized point-location
data structure in the above algorithm, we can use any other data
structure, such as the one of
Edelsbrunner~\etal~\cite{edelsbrunner86}, as long as all predicates\footnote{For instance,  in the case of Mulmuley's 
point-location data structure~\cite{m-fppa-90}, the predicates
(i) and (ii) mentioned at  the beginning of Section~\ref{sec:mixed}.}
used in the associated point-location algorithm are answered by
evaluating the signs of polynomial expressions of bounded degree in
the input data.  If one predicate is answered by evaluating the sign
of several polynomial expressions, it can be split into several
predicates, each of which corresponds to exactly one polynomial
expression. Then, a predicate corresponds to a polynomial expression
of bounded degree that depends on the~$x$ and~$y$ coordinates of the
query point and a set of other parameters (for instance, the coordinates of a point, or the coefficients of an implicit equation
of a curve, against which the query point is tested); this expression, seen as a
polynomial in~$x$ and~$y$, defines a curve,~$\gamma$, of bounded
degree. Recall now that we perform a point location query using an
unknown query point~$s\s$ that lies on a straight or parabolic
arc. The curve~$\gamma$ associated with a polynomial predicate splits
this arc into a bounded number of pieces along which the sign of the
polynomial is constant. To answer the query, it suffices to use the
method described above on the boundary points of these pieces.

\section{Concluding remarks}
\label{sec:conclusion}

We have considered, in this paper, farthest-site Voronoi diagrams of
$k$ disjoint connected polygonal sites in general position and of
total complexity~$n$. In particular, we proved that such diagrams have
complexity $O(n)$ and that they can be computed in $O(n \log^{3} n)$
time.

We have seen that Voronoi regions can consist of several unbounded
components.  However, since the pattern $PQPQ$ cannot appear at
infinity, it is always possible to connect the components of one
Voronoi region by drawing non-crossing connections ``at infinity,''
and so we can think about Voronoi regions as being connected at
infinity.  This curious property can perhaps be better understood by
studying the same problem on the sphere.  Here the resulting structure
is simpler: The bisector of two polygons is a single closed curve, and
the family of bisectors of a fixed polygon~$P$ with the other polygons
forms a collection of pseudo-circles. (The closest-site Voronoi
diagram of three disjoint polygons cannot have three vertices.)  The
farthest-site Voronoi region of~$P$ is the intersection of the
pseudo-disks that do not contain it, and is thus either empty or
simply connected~\cite{m-uevkt-56}. The $O(k)$ bound on the number of
pure vertices is then a simple consequence of the planarity of this
diagram.  (With some care, an alternate proof of
Lemma~\ref{lem:bounded-connected} based on this pseudo-disk property
could be given.)

Farthest-site Voronoi diagrams are related to the function\hfill\break
$f_F:
x\mapsto\arg\max_{S\in\SITES}{\left(\min_{y\in S}{d(x,y)}\right)}$
which returns the farthest site to a query point. The standard
closest-site Voronoi diagram corresponds to the function
$x\mapsto\arg\min_{S\in\SITES}{\left(\min_{y\in S}{d(x,y)}\right)}$
which returns the closest site to a query point. Voronoi diagrams
induced by other similar functions have also been considered in the
literature.  In particular, the diagram obtained by considering the
``dual'' function of $f_F$, that is
$x\mapsto\arg\min_{S\in\SITES}{\left(\max_{y\in S}{d(x,y)}\right)}$ is
the so-called Hausdorff Voronoi diagram; see \cite{papadopoulou04} and
references therein.
%It has been extensively studied by
%Papadopoulou in the  context of its applications to detecting weak
%areas of VLSI designs~\cite{papadopoulou04}.
%\SL{I removed here the comment about the application of Hausdorff Voronoi diagrams to VLSI. It  is interesting but I
%find that it somehow distracts from the actual point since we do not mention any application of the
%other Voronoi diagrams here.} 

%The computation of the farthest-site Voronoi Diagram is related to the
%function $x\mapsto\arg\max_{S\in\SITES}{\left(\min_{y\in S}{d(x,y)}\right)}$
%that computes the site farthest to a query point. The diagram that one obtains
%by considering the ``dual'' function
%$x\mapsto\arg\min_{S\in\SITES}{\left(\max_{y\in S}{d(x,y)}\right)}$ is the
%so-called Hausdorff Voronoi diagram. It has been extensively studied by
%Papadopoulou in the important context of its applications to detecting weak
%areas of VLSI designs~\cite{papadopoulou04}.

A Hausdorff diagram is typically defined for a collection of sets of
points in convex position because the maximum distance from a point to
a polygonal site is realized at a vertex of the polygon's convex
hull. Papadopoulou showed that the size of the Hausdorff Voronoi
diagram is $\Theta(n + M)$, where $n$ is the number of points in the
collection and $M$ is the number of so-called ``crucial supporting
segments'' between pairs of ``crossing sets'' (a pair of sets is
crossing if the convex hull boundary of their union admits more than
two ``supporting segments'', that is, segments joining the convex
hulls of each set; such a segment is said crucial if it is enclosed in
the minimum enclosing circle of each set.)

A fast parallel algorithm was obtained by Dehne~\etal for the special case
when $M=0$~\cite{DehneMT06}. They gave an $O((n \log^4 n)/p)$ time
parallel algorithm for the diagram construction on $p$ processors and a $O(n
\log^4 n)$ time sequential algorithm. Their algorithm is similar to ours and
differs mainly is the way the purple chains are constructed. Indeed, the arcs
of a Hausdorff Voronoi diagram are (straight) line segments; this permits the
use of an ad hoc data structure for finding the ``mixed'' vertices. In
contrast, arcs in a farthest-polygon Voronoi diagram can be curved and we
offer a technique for computing the purple chains using parametric
search, which is more efficient and more general.
This generality has already found another application in the
construction of farthest-site Voronoi diagrams for the geodesic
distance~\cite{Bae09}. We also believe that our technique could be applied in
the sequential algorithm of Dehne~\etal to improve its time complexity to
$O(n\log^3 n)$.

On the other hand, we do not provide a parallel algorithm for
computing farthest-polygon Voronoi diagrams. It would be of 
interest to study the feasibility of applying the parallel techniques of
Dehne~\etal to the farthest-site Voronoi diagram computation.

\section*{Acknowledgments}

We thank the participants of the 8th Korean Workshop on Computational
Geometry, organized by Tetsuo Asano at JAIST, Kanazawa, Japan,
Aug.~1--6, 2005.

\bibliographystyle{geobook}
\bibliography{biblio}

\end{document}